\documentclass[reqno, 11pt]{amsart}
\usepackage{amssymb, natbib}

\usepackage[pdftex]{graphicx}
\usepackage{listings}
\usepackage{multirow}
\usepackage{placeins}
\usepackage{color}
\usepackage{subfigure}
\usepackage{lscape}
\usepackage{dsfont}

\newcommand{\BlackBoxes}{\global\overfullrule5pt}

\BlackBoxes

\newcommand{\R}{\mathbb{R}}
\newcommand{\N}{\mathbb{N}}

\newcommand{\Eop}{\mathbb{E}}
\newcommand{\Pop}{\mathbb{P}}

\newcommand{\FC}{\mathcal{F}}
\newcommand{\h}{{\text H}}

\newtheorem{theorem}{Theorem}

\newtheorem{lemma}[theorem]{Lemma}

\theoremstyle{definition}

\newtheorem{remark}[theorem]{Remark}

\numberwithin{equation}{section} \numberwithin{theorem}{section}

\def\0{\kern0pt\-\nobreak\hskip0pt\relax}

\makeatletter
\AtBeginDocument{%
 \def\@serieslogo{%
 \vbox to\headheight{%
 \parindent\z@ \fontsize{6}{7\p@}\selectfont
 \vss}}}

\def\makeoverbar#1#2#3#4#5#6#7{%
 \setbox0=\hbox{$\m@th#2\mkern#5mu{{}#3{}}\mkern#6mu$}%
 \setbox1=\null \dimen@=#4\fontdimen8#13 \dimen@=3.5\dimen@
 \advance\dimen@ by \ht0 \dimen@=-#7\dimen@ \advance\dimen@ by \wd0
 \ht1=\ht0 \dp1=\dp0 \wd1=\dimen@
 \dimen@=\fontdimen8#13 \fontdimen8#13=#4\fontdimen8#13
 \rlap{\hbox to \wd0{$\m@th\hss#2{\overline{\box1}}\mkern#5mu$}}
 \fontdimen8#13=\dimen@}

\def\mylabel#1#2{{\def\@currentlabel{#2}\label{#1}}}

\makeatother

\begin{document}


\makeatletter \providecommand\@dotsep{5} \makeatother

\title[Portfolio Optimization in Fractional and Rough Heston Models]{Portfolio Optimization in Fractional \\and Rough Heston Models}

\author[N. \smash{B\"auerle}]{Nicole B\"auerle${}^*$}
\address[N. B\"auerle]{Department of Mathematics,
Karlsruhe Institute of Technology (KIT), D-76128 Karlsruhe, Germany}

\email{nicole.baeuerle@kit.edu}

\author[S. \smash{Desmettre}]{Sascha Desmettre${}^\dagger$,$^\ddagger$ }
\address[S. Desmettre]{Department of Mathematics,
TU Kaiserslautern (TUK), D-67663 Kaiserslautern, Germany and Institute for Mathematics and Scientific Computing, University of Graz, Heinrichstra\ss{}e 36, AT-8010 Graz, Austria}

\email{desmettre@mathematik.uni-kl.de ; sascha.desmettre@uni-graz.at}

\thanks{${}^*$ Department of Mathematics,  
Karlsruhe Institute of Technology (KIT), D-76128 Karlsruhe, Germany}
\thanks{${}^\dagger$ Department of Mathematics, TU Kaiserslautern (TUK), D-67663 Kaiserslautern, Germany}
\thanks{${}^\dagger$ Institute for Mathematics and Scientific Computing, University of Graz, AT-8010 Graz, Austria}

\begin{abstract}
We consider a fractional version of the Heston volatility model which is inspired by \cite{GJR14}. Within this model we treat portfolio optimization problems for power utility functions. Using a suitable representation of the fractional part, followed by a reasonable approximation we show that it is possible to cast the problem into the classical stochastic control framework. This approach is generic for fractional processes. We derive explicit solutions and obtain as a by-product the Laplace transform of the integrated volatility. In order to get rid of some undesirable features we introduce a new model for the rough path scenario which is based on the Marchaud fractional derivative. We provide a numerical study to underline our results.
\end{abstract}
\maketitle

\vspace{0.5cm}
\begin{minipage}{14cm}
{\small
\begin{description}
\item[\rm \textsc{ Key words} ]
{\small Fractional stochastic processes; Heston model; Rough paths; Stochastic control; Hamilton-Jacobi-Bellman equation; Feynman-Kac respresentation}
\end{description}
}
\end{minipage}

\section{Introduction}\label{sec:intro}
The stochastic volatility model of \cite{H93} is nowadays a standard model for the pricing of financial derivatives which is underlined by the tremendous amount of related literature; compare e.g. the extensive textbook \cite{R13} and the references therein for a comprehensive overview. In the context of continuous-time portfolio optimization, which is concerned with finding a trading strategy that maximizes expected utility from terminal wealth, the Heston model has among others been dealt with in \cite{Z01,CV05,K05,L01,BL13} using stochastic control methods and in \cite{KMK10} using martingale methods.

Fractional variants of the Heston model, which use a fractional Brownian motion with Hurst index $H >1/2$ as driver of the volatility process, and thus modeling a long term memory effect, have been studied by a large strand of literature, including e.g. \cite{CCR12,LM16}.

Initiated by the observation that volatility is rough in \cite{GJR18}, the current literature takes a new point of view: Rough Heston models, which use a fractional Brownian motion with Hurst index $H <1/2$ as driver of the volatility process, incorporating a better fit of implied volatility surfaces as shown in \cite{GJR18}, have become very popular; compare e.g. \cite{GJR14,ER16}. Subsequently many papers concerning option pricing, simulation of paths, asymptotics, and the foundations of fractional and rough environments have emerged; compare \cite{BFG16,ER17,HJT18,NS16,FZ16,Gea18} to name a few. This increasing importance of rough path theory in general is also reflected by the well-known monographs \cite{LQ02,FH14}.

Portfolio optimization in fractional and rough models, has on the other hand for a long time gained little attention. In an early work, \cite{SVZ07} deals with the Merton problem in a fractional Black-Scholes market. However, with the increasing popularity of rough volatility models, stochastic control methods and portfolio optimization in these models has recently been addressed as follows: For instance, in \cite{DFG17} it is shown for a class of controlled differential equations driven by rough paths that the value function satisfies a Hamilton-Jacobi-Bellman type equation. In a concrete optimal portfolio setting, \cite{FH18a,FH18b} use martingale distortion representations of the value function to establish a first-order approximation of the optimal value, when the return and volatility of the underlying asset are functions of a fractional Ornstein-Uhlenbeck process. 

In this paper we use the classical stochastic control approach and solve the optimal portfolio problem of an investor with a power-utility function in fractional and rough Heston models. Of course a direct application of the stochastic control method is not possible since the respective stochastic processes (volatility and stock price process) are non-Markovian; compare e.g. \cite{N06}.  However, we show by means of a suitable representation of the fractional part followed by a reasonable approximation that it is possible to cast the problem into the classical framework. 
Our calculations are therefore based on a finite dimensional approximation of the underlying volatility process, inspired by the affine representation of fractional processes in \cite{CCM00,HS16}. Solutions to the original optimization problems are then obtained as the limit of the approximated problems. This procedure gives rise to a numerical solution method for these kind of problems. Moreover, as a by-product, we deduce Feynman-Kac type formulas in the fractional and the rough case, which characterize the solutions of the associated partial differential equations as the Laplace transform of the integrated volatility process. In the rough case we use a new model for the volatility which is based on the Marchaud fractional derivative and which seems to remedy some shortcomings of previous models. Indeed it turns out that one has to be very careful with the usage of fractional volatility models for portfolio optimization. This is on one hand due to the general behavior of these models but also due to the dependency properties.

The outline of our paper is as follows: Section~2 introduces the financial market model and the optimization problem in case of a Hurst parameter $H\in(\frac12,1)$. We use here the fractional Riemann-Liouville integral (compare \cite{GJR14}) for the volatility.  In Section~3, we provide a finite dimensional approximation of the optimization problem, derive a solution by solving the corresponding Hamilton-Jacobi-Bellman equation, and verify that the obtained solution is indeed optimal. Section~4 then shows that the solution of the approximated fractional model converges to the solution of the original fractional model. In Section~5 we put emphasis on the definition of a suitable rough Heston model for Hurst parameter $H\in(0,\frac14)$ which is based on the Marchaud fractional derivative and solve the corresponding optimal investment problem. Section~6 then illustrates the fractional and rough Heston models and assesses the behavior of the deduced optimal investment strategies. The appendix contains some proofs and auxiliary results.

\section{The Financial Market Model and the Optimization Problem}\label{sec:mod}\noindent
Suppose that $(\Omega, \FC, (\FC_t)_{ 0\le t\le T}, \Pop)$ is a filtered probability space and $T>0$ is a fixed time horizon. We
consider a financial market with one bond and one risky asset. The bond evolves according to
\begin{equation} d S_t^0 = r S_t^0 dt\end{equation}  with $r>0$ being the
interest rate. The stock price process $S=(S_t)$ is given by
\begin{equation}\label{stockprice} d S_t = S_t \left( (r+\lambda \nu_t) dt + \sqrt{\nu_t}d B_t^S\right) \end{equation}
where $(B^S_t)$ is an $(\FC_t)$-Brownian motion and $\lambda>0$ a constant. The volatility process $(\nu_t)$ is a 'fractional' Cox-Ingersoll-Ross process with  $\nu_0 := v_0\ge 0$:
\begin{equation}\label{vola}
\nu_t = v_0 + \frac{1}{\Gamma(\alpha)} \int_0^t (t-s)^{\alpha-1} Z_sds \end{equation}
where  $\alpha = 2\h -1  \in (0,1)$ with Hurst index $\h \in \left(\tfrac{1}{2},1\right)$ and
\begin{equation}\label{eq:Z}
d Z_t =\kappa (\theta -Z_t) dt + \sigma \sqrt{Z_t} d B_t^Z
\end{equation}
with $Z_0:= z_0\ge 0$ is the usual Cox-Ingersoll-Ross model. The constants $\kappa, \theta, \sigma$ are assumed to be positive and satisfy the Feller condition $2\kappa\theta \ge \sigma^2$. This implies that $(Z_t)$ stays strictly positive with probability one. The 'rough volatility' case $H\in(0,\frac14)$ will be considered later. Also $(B^Z_t)$ is an $(\FC_t)$-Brownian motion. We assume that $(B^S_t)$ and $(B^Z_t)$ are correlated  with correlation $\rho\in (-1,1)$, i.e. $\langle B^S,B^Z\rangle_t = \rho t$. Note that the integral which appears in \eqref{vola} is well-defined for $\alpha  \in (0,1)$ since $\Gamma(z) := \int_0^\infty e^{-t} t^{z-1}dt$ is well-defined for $z>0$. The operator
\begin{equation}
I_{0}^\alpha f(t) := \frac{1}{\Gamma(\alpha)}\int_0^t (t-s)^{\alpha-1} f(s) ds
\end{equation}
is the classical left fractional Riemann-Liouville integral of order $\alpha$, also called Euler transformation (see Definition 2.1 in \cite{SKM93}). Among others it has the property that
\begin{equation}
\lim_{\alpha\to 0} I_{0}^\alpha f(t) = f(t)
\end{equation}
pointwise (see e.g. Theorem 2.7. in \cite{SKM93}) which means that in the limiting case $\alpha\downarrow 0$ we obtain the classical Heston model of~\cite{H93}.
 Since $Z_t \ge 0$ almost surely we obtain that $\nu_t \ge \nu_0$ almost surely for all $t\ge 0$. It can be shown (see \cite{GJR14}) for $t,h\ge 0$ that 
\begin{equation}\label{eq:long_range_1}
Cov(\nu_{t+h},\nu_t) = \frac{1}{\Gamma^2(\alpha)} \int_0^{t+h}\int_0^t(t-s)^{\alpha-1}(t+h-u) ^{\alpha-1} Cov(Z_s,Z_u)dsdu
\end{equation}
where 
\begin{equation}\label{eq:long_range_2}
Cov(Z_s,Z_u) = \sigma^2 \Big(\frac{\theta}{2\kappa}  e^{-\kappa |s-u|}+ \frac{z_0-\theta}{\kappa} e^{- \kappa(s\wedge u)} -\frac1{2\kappa} (2z_0-\theta)e^{-\kappa(s+u)} \Big)
\end{equation}
and $s\wedge u = \min(s,u)$.
This implies that the volatility process $(\nu_t)$ possesses long-range dependence. Moreover, the operator $I_0^\alpha$ has a smoothing property (for simulation results see Section~\ref{sec:simulation}).

\begin{remark}
\cite{ER17} give an alternative formulation of a fractional/rough Heston model, which is closer to the Mandelbrot-van Ness representation of fractional Brownian motion than the fractional CIR process given by (2.3) and (2.4). We opted for our formulation close to \cite{GJR14}, as it turned out that the Hamilton Jacobi Bellmann equations corresponding to this problem are more tractable as the ones which correspond to the model in \cite{ER17}.
\end{remark}

 Using the fact that for $\alpha\in (0,1)$
\begin{equation}\label{eq:Gamma} \frac{(t-s)^{\alpha-1}}{\Gamma(\alpha)} = \int_0^\infty e^{-(t-s)x} \mu(dx),\quad\mbox{with  } \mu(dx)= \frac{dx}{x^\alpha \Gamma(\alpha)\Gamma(1-\alpha)}
\end{equation}
we obtain with the Fubini Theorem 
\begin{eqnarray*}
\nonumber \nu_t &=& v_0 +\int_0^t \int_0^\infty e^{-(t-s)x} Z_s \mu(dx)ds \\
\nonumber &=& v_0 +\int_0^\infty  \int_0^t  e^{-(t-s)x} Z_s ds  \mu(dx)   \\ \label{eq:Vtrepresentation}
   &=& v_0 +\int_0^\infty  Y_t^x  \mu(dx),
\end{eqnarray*}
where 
\begin{equation}Y_t^x := \int_0^t  e^{-(t-s)x} Z_s ds.
\end{equation}

Using partial integration we see that $(Y_t^x)$ satisfies the stochastic differential equation
\begin{equation}
dY_t^x =  (Z_t- x Y_t^x) dt.
\end{equation}

The optimization problem is to find self-financing investment strategies in this market that maximize the expected utility from terminal wealth.
As utility function we choose the power utility function $U(x)= \frac1\gamma x^\gamma$ with $\gamma<1,\gamma \neq 0$. The parameter $\gamma$ represents the risk aversion of the investor. Smaller $\gamma$ correspond to higher risk aversion. In what follows we denote by $\pi_t\in\R$ the {\em fraction of wealth} invested in the stock at time $t$. $1-\pi_t$ is then the fraction of wealth invested in the bond at time $t$. If $\pi_t<0$, then this means that the
stock is sold short and $\pi_t>1$ corresponds to a credit. The process $\pi=(\pi_t)$ is called {\em portfolio strategy}. An admissible portfolio strategy has to be an $(\FC_t)$-adapted process such that all integrals exist. The {\em wealth process} under an admissible portfolio strategy $\pi$ is given by the solution of the stochastic differential
equation
\begin{eqnarray}\label{eq:wealth_process}
d W_t^\pi &=& W_{t}^\pi (r+ \pi_t \lambda \nu_t)dt + W_t^\pi\pi_t \sqrt{\nu_t}d B_t^S,
\end{eqnarray} where we assume that ${W}_0=w_0>0$ is the given initial
wealth. The {\em optimization problem } is  defined by
\begin{eqnarray}\label{eq:optprob}
&& V(w_0,v_0,z_0) := \sup_{\pi} \Eop_{w_0,v_0,z_0}\left[\frac1\gamma\big({W}_T^\pi\big)^\gamma\right]
\end{eqnarray}
where $\Eop_{w_0,v_0,z_0}$ is the conditional expectation given $W_0= w_0, \nu_0=v_0, Z_0=z_0$ and the supremum is taken over all admissible portfolio strategies. A portfolio strategy $\pi^\ast$ is {\em optimal} if it attains the supremum. Seen as an optimization problem with state process $(W_t^\pi)$ this problem is non-Markovian and the standard stochastic control approach cannot be applied. By the alternative representation of the volatility in   \eqref{eq:Vtrepresentation}, the problem can be 'Markovianized', however  only with infinite dimensional process $(W_t,Y_t^x,Z_t)$ where $x>0$. Thus, we solve the problem by first looking at a finite dimensional approximation.

\begin{remark}
Of course it is reasonable to assume that the stock price process is observable for the decision maker. By observing $(S_t)$ we are also able to observe the quadratic variation
\begin{equation}
\left<S\right>_t = \int_0^t S_u^2 \nu_u du,
\end{equation}
which implies that $(\nu_t)$ is observable. Using \eqref{vola} this implies that $(Z_t)$ is observable which defines $(Y_t^x)$. Thus it is realistic to assume the knowledge of $\nu_t, Y_t^x$ and $Z_t$ in this model.  
\end{remark}

\section{A finite dimensional Approximation of the Optimization Problem}\label{sec:approx}\noindent
The idea is to consider the representation 
\begin{equation}
 \nu_t =  v_0 +\int_0^\infty  Y_t^x  \mu(dx), \quad \mbox{ with } \quad Y_t^x := \int_0^t  e^{-(t-s)x} Z_s ds
\end{equation}
of the fractional volatility process $(\nu_t)$ given by the dynamics 
\begin{equation}
\nu_t = v_0 + \frac{1}{\Gamma(\alpha)} \int_0^t (t-s)^{\alpha-1} Z_s ds \quad \mbox{ with } \quad d Z_t =\kappa (\theta -Z_t) dt + \sigma \sqrt{Z_t} d B_t^Z\,,
\end{equation}
and to approximate $\mu$ by a discrete measure with a finite number of atoms. Therefore we use a quantization of $\mu$ which is defined as follows (this has also been used in \cite{CCM00}): 

Let $\mathcal{Z}^{n} := \{0<\xi_0^{n}<\ldots <\xi_n^{n}<\infty\}$ and define the barycenter of $\mu$ on the respective intervals $(\xi_i^{n},\xi_{i+1}^{n})$ for $ i=0,\ldots,n-1$ by
\begin{equation}
x_{i+1}^{n} := \frac{\int_{\xi_i^{n}}^{\xi_{i+1}^{n}}x\mu(dx)}{\int_{\xi_i^{n}}^{\xi_{i+1}^{n}}\mu(dx)},
\end{equation} 
and the mass on the atom for $ i=0,\ldots,n-1$ by
\begin{equation}
q_{i+1}^{n} := \int_{\xi_i^{n}}^{\xi_{i+1}^{n}}\mu(dx).
\end{equation}
The corresponding measure is given by 
\begin{equation}\label{eq:mumeasure}
\mu^{n} := \sum_{i=1}^n q_i^{n} \delta_{x_i^{n}}
\end{equation} 
where $\delta_x$ is the Dirac measure on $x$. In what follows we assume that $\mathcal{Z}^{n} $ satisfies:
\begin{itemize}
\item[(i)] $\xi_0^{n} \to 0$ and $\xi_n^{n} \to \infty$ for $n\to\infty$,
\item[(ii)] $\delta(\mathcal{Z}^{n} ) := \max_{i=0,\ldots,n-1} |\xi_{i+1}^{n}-\xi_{i}^{n}| \to 0$ for $n\to\infty$,
\item[(iii)] $\mathcal{Z}^{n}\subset \mathcal{Z}^{n+1}$.
\end{itemize}
With these definitions we obtain:

\begin{lemma}\label{lem:conv1} Suppose $f\in L^1(\mu)$ and the sequence $(\mathcal{Z}^{n})$ satisfies (i)-(iii) above. Then
\begin{equation}
 \int f d \mu^{n} \to \int fd\mu, \quad n\to \infty\end{equation}
 if  $f\ge 0$ and $f$ is convex. Moreover, the convergence is monotone increasing.
\end{lemma}

The proof of this lemma and all other longer proofs are deferred to the appendix.
Now let us denote the finite dimensional approximate volatility process by 
\begin{equation}\label{eq:nutn}
\nu_t^n := v_0 + \int_0^\infty  Y_t^x  \mu^n(dx) =  v_0+\sum_{i=1}^{n} q_i^n Y_t^{x_i^n}
\end{equation}
where $\mu^n$ is defined in \eqref{eq:mumeasure}.
Then a direct application of the previous  lemma implies now the following result:

\begin{theorem}\label{theo:conv}
Under the assumptions of Lemma \ref{lem:conv1} it holds for $t\ge 0$ that
\begin{equation}
\nu_t^n  \uparrow \nu_t  
\end{equation} 
for $n\to\infty$ almost surely.
\end{theorem}

\begin{proof}
We have to show that 
\begin{equation}\int_0^\infty  Y_t^x  \mu^n(dx) \to  \int_0^\infty  Y_t^x  \mu(dx)\end{equation}
where $Y_t^x = \int_0^t e^{-(t-s)x} Z_s ds$. Obviously for fixed $\omega\in \Omega$ and fixed $t>0$, the function $x\mapsto Y_t^x(\omega)$ is non-negative and convex (note that $Z_s(\omega)$ is for all $s>0$ positive).  Thus Lemma \ref{lem:conv1}  implies that $\nu^n_t \uparrow \nu_t$ for $n\to\infty$ almost surely.
\end{proof}

Instead of $(\nu_t)$ from \eqref{vola} we consider now $(\nu_t^n)$ from \eqref{eq:nutn}.
 All stochastic processes in this section depend  on $n$ (with the exception of $(Z_t)$) but in order to ease notation we do not make this dependence  explicit in the notation. Thus the dynamics of the approximate stock price process are given by
\begin{equation}\label{stockprice2} 
d S_t = S_t \Big( (r+\lambda \Big(v_0+\sum_{i=1}^{n} q_i Y_t^{x_i})\Big) dt + \sqrt{\Big(v_0+\sum_{i=1}^{n} q_i Y_t^{x_i}\Big)}d B_t^S\Big) \end{equation}
where as before for $x>0$ we have
\begin{eqnarray}
 dY_t^x &=& (Z_t- x Y_t^x) dt\\
d Z_t &=&  \kappa(\theta -Z_t) dt + \sigma \sqrt{Z_t} d B_t^Z.
\end{eqnarray}
The stochastic differential equation for the approximate wealth process is thus
\begin{equation}\label{eq:wealth2}
d W_t^\pi =W_{t}^\pi \Big(r+ \pi_t \lambda \Big(v_0+\sum_{i=1}^{n} q_i Y_t^{x_i}\Big)\Big)dt + W_t^\pi \pi_t \sqrt{\Big(v_0+\sum_{i=1}^{n} q_i Y_t^{x_i}\Big)}d B_t^S.
\end{equation}

We consider the same optimization problem as in \eqref{eq:optprob} with the preceding processes. This results in a finite dimensional classical stochastic optimal control problem. We have to consider the value functions
\begin{equation}\label{eq:optprobapprox}V(t,w,y_1,\ldots,y_n,z) := \sup_{\pi} \Eop_{t,w,y_1,\ldots,y_n,z}\left[\frac1\gamma\big({W}^{\pi}_T\big)^\gamma\right].
\end{equation}
where $\Eop_{t,w,y_1,\ldots,y_n,z}$ is the conditional expectation given $W_t=w, Y_t^{x_i}=y_i, Z_t=z$ at time $t$. As before portfolio strategies are $(\FC_t)$-adapted processes. In what follows we derive the corresponding Hamilton-Jacobi-Bellman (HJB) equation for this optimization problem. We denote the generic function by $G(t,w,y_1,\ldots,y_n,z)$ with $t\in[0,T], w> 0, y_i \ge 0, z> 0$. The boundary condition is given by $G(T,w,y_1,\ldots,y_n,z)= \frac 1\gamma w^\gamma$. In order to ease notation, we set $\beta := v_0+\sum_{i=1} ^{n}  q_i y_i$. Thus, the HJB equation reads
\begin{eqnarray}\label{eq:HJB1}
 \nonumber 0 = \sup_{u\in\R}&& \Big\{G_t +G_w w(r+u\lambda \beta)  + \sum_{i=1}^n G_{y_i} \big(z-x_iy_i \big) +G_z \kappa(\theta-z) \\
   &&+ \frac12 G_{ww}w^2 u^2 \beta  + \frac12 G_{zz} \sigma^2 z +G_{wz} w u\sigma \rho  \sqrt{z\beta}\Big\}.
\end{eqnarray}

We will first show that a classical solution of this HJB equation exists and can be given explicitly in the uncorrelated case $\rho=0$.

\begin{theorem}\label{theo:HJB}
A solution of HJB equation \eqref{eq:HJB1} exists on a certain time interval $[0,T_\infty]$ and is for $t\in[0,T], T\le T_\infty, w>0, y_i>0, z>0$ given by
\begin{equation}G(t,w,y_1,\ldots,y_n,z)= \frac1\gamma w^\gamma  g(t,y_1,\ldots,y_n,z)^c \end{equation} with $c= \frac{1-\gamma}{1-\gamma+\gamma\rho^2}$ and a differentiable $g$ satisfying 
\begin{equation}\label{eq:g}
0= cg_t +g\Big(\gamma r+\frac12 \frac{\lambda^2\beta\gamma}{1-\gamma}\Big)+c  \sum_{i=1}^n g_{y_i} (z-x_iy_i) + cg_z \Big( \kappa(\theta-z)+\frac{\lambda\gamma\sigma\rho\sqrt{z\beta}}{1-\gamma}\Big)+\frac12 \sigma^2 czg_{zz}.
\end{equation} In the uncorrelated case $\rho=0$, the function  $g$ can be given explicitly by
$$ g(t,y_1,\ldots,y_n,z) =\exp\Big(\phi(T-t)+\sum_{i=1}^n \psi_i(T-t)y_i+\varphi(T-t)z \Big)  $$
where 
\begin{equation}\label{eq:psii}
\psi_i(T-t) =  \eta q_i \int_0^{T-t}e^{ - x_is} ds
\end{equation}
with $\eta := \frac12\frac{\gamma\lambda^2}{1-\gamma}$ and $\varphi$ and $\phi$ are solutions of the ordinary differential equations
\begin{eqnarray}
 \label{riccati1}
   \varphi_t(T-t) &=& \eta \int_0^{T-t} \int_0^\infty e^{-xs} \mu^{n}(dx)ds -\kappa\varphi(T-t)+ \frac12\sigma^2 \varphi^2(T-t) \\ \label{ode2}
    \phi_t(T-t) &=& \gamma r+v_0\eta+\varphi (T-t)\kappa\theta.
\end{eqnarray}
with boundary condition $\varphi(0)=\phi(0)=0$.
\end{theorem}

\begin{remark}\label{rem:sol}
Note that \eqref{riccati1}  is a Riccati equation and \eqref{ode2} can be solved explicitly once the solution for $\varphi$ is known. 
\end{remark}

For the following discussion it is important to note that the solution $g$  of the partial differential equation  \eqref{eq:g} can be represented as a Laplace Transform of an integrated  volatility via the Feynman-Kac Theorem.

\begin{theorem}\label{theo:FK}
A solution $g$ of the partial differential equation \eqref{eq:g} with boundary condition $g(T,y_1,\ldots,y_n,z) =1$  can for $t\in[0,T], T\le T_\infty, w>0, y_i>0, z>0$  be written as
\begin{equation}\label{eq:gFK} g(t,y_1,\ldots,y_n,z) = \Eop_{t,y_1,\ldots,y_n,z}\left[ e^{\int_t^T \Big( \frac{\gamma r}{c}+ \frac{1}{2}\frac{\gamma \lambda^2}{(1-\gamma)c} \tilde{\nu}_s^n  \Big) ds} \right]
\end{equation}
where the constant $c$ has been defined in the previous theorem, $\tilde{\nu}_t^n = v_0+ \sum_{i=1}^{n} q_i^n \tilde{Y}_t^{x_i^n}$ and
\begin{eqnarray}
 d\tilde{Y}_t^x &=& (\tilde{Z}_t^n- x \tilde{Y}_t^x) dt\\ \label{eq:tildeZ}
d \tilde{Z}_t^n &=&  \Big(\kappa(\theta -\tilde{Z}_t^n) +\frac{\lambda \gamma\sigma\rho}{1-\gamma}\sqrt{\tilde{Z}_t^n\tilde{\nu}_t^n}\Big)dt + \sigma \sqrt{\tilde{Z}_t^n} d B_t^Z.
\end{eqnarray}
In the uncorrelated case $\rho=0$ we have $\nu_t = \tilde{\nu}_t$
\end{theorem}

\begin{proof}
This statement follows from Theorem 1 of \cite{HS00}  where we identify the components  $X^i$ of process $X$ with the processes $\tilde{Z}, \tilde{Y}^j, j=1,\ldots,n$. The function $c$ in front of $f$ which appears in \cite{HS00} is given by \begin{equation}c(t,y_1,\ldots, y_n,z)=\frac{\gamma r}c+ \frac{1}{2}\frac{\gamma \lambda^2}{(1-\gamma)c} \big(v_0+ \sum_{i=1}^n y_i q_i\big).\end{equation}
We also have 
\begin{equation}b^{n+1}(t,y_1,\ldots, y_n,z)= \kappa(\theta-z) + \frac{\lambda\gamma\sigma\rho}{1-\gamma} \sqrt{z(v_0+\sum q_i y_i)} \end{equation}
Moreover the functions $g$ and $h$ in \cite{HS00} are here given by $g\equiv 0, h\equiv 1$. Note also that (A1), (A2) and (A3') of Theorem 1 in \cite{HS00}  are  satisfied.
\end{proof}

\begin{theorem}\label{theo:valuefinite}[Verification]
Suppose that $G(t,w,y_1,\ldots,y_n,z) := \frac1\gamma w^\gamma g(t,y_1,\ldots,y_n,z)^c $ with $g$ as in \eqref{eq:gFK}. Then for $t\in[0,T], T\le T_\infty$, an optimal investment strategy $(\pi_t^*)$ for problem \eqref{eq:optprobapprox} is given by 
\begin{align}\label{eq:pi_star}
\pi_t^*=\frac{\lambda}{1-\gamma}+ \frac{c\sigma\gamma}{1-\gamma}\sqrt{\frac{Z_t}{\nu_t^n}}\frac{g_z}{g}\,,
\end{align}
and $V=G$, i.e.\ $G$ coincides with the value function provided that $\int_0^t g^c W_s^\gamma \pi_s^* \sqrt{\nu_s^n}  \,dB_s^S$ and $ \int_0^t G_{z} \sigma \sqrt{Z_s} \,dB_s^Z$ are true martingales. In the uncorrelated case $\rho=0$ we have $\pi_t^*\equiv \frac{\lambda}{1-\gamma}$  and the value function can be written as
\begin{equation}\label{eq:Vexpression}
V(t,w,y_1,\ldots,y_n,z) = \frac1\gamma w^\gamma \exp\Big(\phi(T-t)+\sum_{i=1}^n \psi_i(T-t)y_i+\varphi(T-t)z \Big)\,, 
\end{equation}
where $\psi_i$ are given in \eqref{eq:psii} and $\varphi$ and $\phi$ are solutions of \eqref{riccati1} and \eqref{ode2}.
\end{theorem}

\begin{remark}\label{rem:Merton_ratio}
Note that in case $\rho=0$ the optimal portfolio strategy does not depend on the volatility at all. This is typical in settings where the Brownian motions of stock and volatility process are uncorrelated and where the appreciation rate and the volatility are in a certain relation (see e.g. \cite{BL13}). It corresponds to the Merton ratio which would exactly be $\lambda$ in our model.
\end{remark}

\section{The Fractional Optimization Problem}
We solve now problem \eqref{eq:optprob} by taking the limit $n\to\infty$ in the results of the previous section.  The special case $\rho=0$ will be discussed separately in section 4.2.

\subsection{The Correlated Case}
From Theorem \ref{theo:FK} we know that the value function of the approximation is essentially given by a Laplace transform of an integrated process $(\tilde{\nu}_t^n)$. We first consider $(\tilde{Z}_t^n)$ given in  \eqref{eq:tildeZ} and discuss what happens if $n$ tends to $\infty$. For this purpose we write it as 
\begin{eqnarray}
\tilde{Z}_t^n = \tilde{Z}_0 + \!\int_0^t  \! \!\kappa(\theta-\tilde{Z}_s^n) + \frac{\lambda\gamma\sigma\rho}{1-\gamma} \sqrt{\tilde{Z}_s^n} \sqrt{v_0\! +\! \sum_{i=1}^n q_i^n \int_0^s \! \! e^{-(s-u)x_i}\tilde{Z}_u^ndu} ds + \sigma\!\int_0^t  \! \!\sqrt{\tilde{Z}_s^n}dB_s^Z.
\end{eqnarray}

For the next result we consider $(\tilde{Z}_t^n)$  as a random element in $D_{\R}[0,T_\infty]$ and denote by $\Rightarrow$ weak convergence. 

\begin{lemma}\label{lem:weakconv}
It holds that $(\tilde{Z}_t^n) \Rightarrow (\tilde{Z}_t)$ in Skorohod topology and $ (\tilde{Z}_t)$  satisfies
\begin{eqnarray}
\tilde{Z}_t = \tilde{Z}_0 + \!\int_0^t  \! \!\kappa(\theta-\tilde{Z}_s) + \frac{\lambda\gamma\sigma\rho}{1-\gamma} \sqrt{\tilde{Z}_s} \sqrt{v_0\! +\! \int_0^\infty \! \! \int_0^s  \! \!e^{-(s-u)x}\tilde{Z}_udu\mu(dx)} ds + \sigma\!\int_0^t  \! \!\sqrt{\tilde{Z}_s}dB_s^Z. 
\end{eqnarray}
\end{lemma}

Hence we have  convergence of $(\tilde{Z}_t^n)$ to a limit $ (\tilde{Z}_t)$ which satisfies the SDE we get by replacing the sum by the integral. We proceed with the wealth process itself.

The stochastic differential equation \eqref{eq:wealth2} can be solved explicitly and we obtain for an arbitrary admissible portfolio strategy $\pi$
\begin{eqnarray}\label{eq:comp}
  W_T^{\pi,n} &=& w_0 \exp\Big(\int_0^T  \Big(r+\pi_s \nu_s^n(\lambda-\frac12 \pi_s)\Big)ds + \int_0^T \pi_s \sqrt{\nu_s^n}dB_s^S\Big).
\end{eqnarray}
In what follows we denote
\begin{eqnarray}
V_n(w_0,v_0,z_0;\pi)& :=& \Eop_{w_0,v_0,z_0}\Big[\frac{1}{\gamma}(W_T^{\pi,n})^\gamma\Big],\\
V_n(w_0,v_0,z_0)  &:=& \sup_\pi V_n(w_0,v_0,z_0;\pi),\\
V(w_0,v_0,z_0;\pi)  &:= & \Eop_{w_0,v_0,z_0}\Big[\frac{1}{\gamma}(W_T^{\pi})^\gamma\Big],\
\end{eqnarray}
Then it can be shown that

\begin{lemma}\label{lem:conVpi}
Suppose that for fixed strategy $\pi$ the sequence $(W_T^{\pi,n})$ is uniformly integrable. Then $$ \lim_{n\to\infty} V_n(w_0,v_0,z_0;\pi)   = V(w_0,v_0,z_0,\pi).  $$ 
\end{lemma}

The value of $V_n(w_0,v_0,z_0)$ is given explicitly in Theorem \ref{theo:HJB} and Theorem \ref{theo:FK} and the limit $\bar{V}:=\lim_{n\to\infty}V_n$ exists due to monotone convergence (see Lemma \ref{lem:comp}). 
Finally we get

\begin{theorem}\label{theo:epsopt}[$\varepsilon$-optimal strategies]
Suppose the assumptions of Lemma \ref{lem:conVpi} are valid.
Let $\varepsilon>0$ and choose $n$ large enough such that $|\bar{V}-V_n|<\frac{\varepsilon}{2}$ as well as $|V_n(w_0,v_0,z_0;\pi^n)-V(w_0,v_0,z_0;\pi^n)|<\frac{\varepsilon}{2}$ where $\pi^n$ is the optimal strategy for approximation $n$. Then $\pi^n$ is $\varepsilon$-optimal for the original portfolio problem \eqref{eq:optprob}.
\end{theorem}

\begin{proof}
Since
\begin{eqnarray}
\bar{V}&=&\lim_{n\to\infty}\sup_\pi V_n(\cdot;\pi),\\
{V}&=&\sup_\pi\lim_{n\to\infty} V_n(\cdot;\pi),
\end{eqnarray}
we obtain that $\bar{V}\ge V$.  Thus it follows
\begin{equation}
 0\le V-V(\cdot,\pi^n)\le \bar{V}-V(\cdot,\pi^n)\le |\bar{V}-V_n|+|V_n - V(\cdot;\pi^n)|\le \varepsilon
\end{equation}
which implies the statement.
\end{proof}

Hence we can solve the original problem up to an arbitrarily small error by solving the approximate problem. In the uncorrelated case it is possible to solve the original problem exactly. We will do this in the next section.

\subsection{The Uncorrelated Case}
When $\rho=0$ we get more explicit results. We first consider the differential equations \eqref{riccati1} and \eqref{ode2} which define the value function. Note that the second differential equation  does not depend on $n$. The Riccati equation for $\varphi$ formally turns for $n\to\infty$ to 
\begin{equation}
 \varphi_t(T-t) = \eta \int_0^{T-t} \int_0^\infty e^{-xs} \mu(dx)ds-\kappa\varphi(T-t)+\frac12\sigma^2 \varphi^2(T-t)
\end{equation}
with boundary condition $\varphi(0)=0$. Using \eqref{eq:Gamma} we obtain that
\begin{equation}
\int_0^\infty e^{-xs} \mu(dx) = \frac{s^{\alpha-1}}{\Gamma(\alpha)}
\end{equation}
and further
\begin{equation}
\int_0^{T-t} \frac{s^{\alpha-1}}{\Gamma(\alpha)}ds = \frac{(T-t)^\alpha}{\Gamma(\alpha+1)}.
\end{equation}
Finally, in the limiting case, the system of differential equations \eqref{riccati1}, \eqref{ode2} is given by
\begin{eqnarray}
\label{riccati2} \varphi_t(T-t) &=& \eta \frac{(T-t)^\alpha}{\Gamma(\alpha+1)} -\kappa\varphi(T-t)+\frac12\sigma^2 \varphi^2(T-t) \\ \label{ode3}
  \phi_t(T-t) &=& \gamma r+v_0\eta+\varphi (T-t)\kappa\theta
\end{eqnarray} 
with boundary condition $\varphi(0)=\phi(0)=0$. A relation between these differential equations which may be obvious, but has to be shown, is given in the next lemma:  

\begin{lemma}\label{lem:odeconvergence}
Let $\varphi^n$ be the solution of \eqref{riccati1} and $\varphi$ the solution of \eqref{riccati2}. Then it holds that $\varphi^n(t) \to \varphi(t)$ pointwise in $t$ for $n\to\infty$.
\end{lemma}
 
This result can be used to show
 
\begin{theorem}\label{theo:limitcase} [Solution for the fractional path Problem]
The optimal portfolio strategy for problem \eqref{eq:optprob} in case $\rho=0$ is given by $\pi_t^*\equiv \frac{\lambda}{1-\gamma}$  and the value function can be written as
\begin{equation}\label{eq:valuerough}
V(w_0,v_0,z_0) = \frac1\gamma w_0^\gamma \exp\Big(\phi(T)+\varphi(T)z_0 \Big)
\end{equation}
where $\varphi$ and $\phi$ are solutions of \eqref{riccati2} and \eqref{ode3}. In case $\gamma<0$ we have to assume that $(W_T^{\pi,n})$ is uniformly integrable for all $\pi$.
\end{theorem}

\begin{remark}
If we set $v_0=\alpha=0$ in \eqref{riccati2} and \eqref{ode3} then the resulting differential equations are the same as for the classical Heston model with CIR volatility (see e.g. \cite{K05}). This implies that in the limiting case $\alpha\to 0$, the considered model corresponds to the classical Heston model. See also the discussion in Section 2.
\end{remark}

The value function in Theorem \ref{theo:limitcase} gives the maximal expected utility at time point $t=0$. In Theorem \ref{theo:valuefinite}, the optimal value is presented for an arbitrary starting time point $t\in[0,T]$. In this case, the value depends on the history of the realized volatility via the variables $y_i=Y_t^{x_i}$. Of course it is possible to derive from \eqref{eq:Vexpression} the value function in the fractional Heston model by taking $n\to\infty$. In order to do so, we have to consider the additional term $\sum_{i=1}^n \psi_i(T-t)y_i$ which appears in the exponential.
For $n\to \infty$ we obtain
\begin{eqnarray}
\nonumber\lim_{n\to\infty} \sum_{i=1}^n \psi_i(T-t)y_i &=& \lim_{n\to\infty} \eta\sum_{i=1}^n q_i y_i \int_0^{T-t} e^{-x_i s} ds\\
&=& \eta \int_0^\infty Y_t^x \int_0^{T-t} e^{-x s} ds \mu(dx).
\end{eqnarray}
This term can be expressed with the help of $(Z_t)$ as follows:
\begin{eqnarray}
\nonumber \int_0^\infty Y_t^x \int_0^{T-t} e^{-x s} ds \mu(dx) &=&  \int_0^\infty Y_t^x \frac{1}{x}\big( 1- e^{-x(T-t)}\big) \mu(dx)\\
 &=& \int_0^\infty \int_0^t  \frac{e^{ux}}{x}\big(e^{-tx}- e^{-Tx}\big) Z_u du \mu(dx).
\end{eqnarray}
Thus the value function in the fractional Heston model at time $t$, given $W_t=w, (Z_s)_{0\le s\le t} = (z_s)_{0\le s\le t} ,z_t=z,$ can be written as
\begin{equation}\label{eq:value_frac}
\frac1\gamma w^\gamma \exp\Big(\phi(T-t)+\varphi(T-t)z \Big) 
\exp\Big(\eta\int_0^\infty \int_0^t  \frac{e^{ux}}{x}\big(e^{-tx}- e^{-Tx}\big) Z_u du \mu(dx)\Big),
\end{equation}
where as before $\varphi$ and $\phi$ are solutions of \eqref{riccati2}.
Now further note that since
\begin{equation}
e^{-t x} \le  \frac{e^{-Tx}-e^{-tx}}{x(T-t)}\le -e^{-Tx} 
\end{equation}
the last factor in \eqref{eq:value_frac} is bounded from below by
\begin{eqnarray}
\nonumber && \exp\Big(\eta\int_0^\infty \int_0^t  \frac{e^{ux}}{x}\big(e^{-tx}- e^{-Tx}\big) Z_u du \mu(dx)\Big)\ge\exp\Big(\eta(T-t)\int_0^\infty \int_0^t  e^{x(u-T)} Z_u du \mu(dx)\Big)\\
\nonumber &=  & \exp\Big(\eta(T-t)\int_0^t  \int_0^\infty  e^{x(u-T)}  \mu(dx)Z_u du\Big)
=  \exp\Big(\frac{\eta(T-t)}{\Gamma(\alpha)}\int_0^t (T-u)^{\alpha-1}  Z_u du\Big)\\
&\ge& \exp\Big(\frac{\eta(T-t)}{\Gamma(\alpha) T^{1-\alpha}}\int_0^t   Z_u du\Big) > 1
\end{eqnarray}
This shows that given the knowledge about some historical data from $(Z_t)$ increases the value function in this fractional Heston model. This might be expected because the positive dependence of the volatility should be exploited. Indeed, suppose $T-t$ is fixed and both $T$ and $t$ are increasing. Since on average $Z_u \approx \theta$, the factor grows with rate $T^\alpha$, i.e. a larger history yields a larger portfolio value. 


\begin{remark}
It is well-known that the portfolio optimization problem with logarithmic utility function, i.e. $U(x)=\ln(x)$ can be obtained in the limit from the power utility problem  when we let $\gamma\downarrow 0$. Obviously in this case the strategy $\pi_t^* \equiv \lambda$ is optimal. This can be seen by a direct investigation of the optimization problem.
\end{remark}

\section{The Rough Volatility Case $\alpha\in (-1,-\frac12)$}\label{sec:rough}
Let now {$\alpha = 2\h -1 \in (-1,-\frac12)$ with Hurst index $\h \in \left(0,\tfrac{1}{4}\right)$}. In this case we cannot use the definition \eqref{vola} which only makes sense for exponents $\alpha>0$. Also the definition $\frac{d}{dt} I_0^{\alpha+1}f(t)$ cannot be applied since this requires $f$ to be absolutely continuous, however we have to plug in the paths of $(Z_t)$ which are only almost H\"older $\frac12$-regular (see appendix). Instead we use here the so-called {\em Marchaud fractional derivative}. It is for $\alpha\in (-1,-\frac12)$ defined by (see \cite{SKM93}, sec. 13.1)
\begin{equation}
D^{\alpha+1}_0f(t) = f(t) \frac{t^{-\alpha-1}}{\Gamma(-\alpha)}+\frac{\alpha+1}{\Gamma(-\alpha)} \int_0^t \frac{f(t)-f(s)}{(t-s)^{\alpha+2}}ds
\end{equation}
which coincides with the Riemann-Liouville fractional derivative if $f$ is sufficiently differentiable  (see \cite{SKM93}, Section 13) and is defined for H\"older $\delta$-continuous functions with $\alpha+1<\delta\le 1$. Since $\delta<\frac12$ in our application, $\alpha$ has to be from the interval $(-1,-\frac12)$. Thus, the volatility process is now defined  with  $\nu_0 := v_0\ge 0$ by:
\begin{equation}\label{vola2}
\nu_t = v_0 + Z_t  \frac{t^{-\alpha-1}}{\Gamma(-\alpha)} + \frac{\alpha+1}{\Gamma(-\alpha)} \int_0^t \frac{Z_t-Z_s}{(t-s)^{\alpha+2}}ds
 \end{equation}
where  again
\begin{equation}\label{eq:Z2}
d Z_t = \kappa(\theta -Z_t) dt + \sigma \sqrt{Z_t} d B_t^Z
\end{equation}
with $Z_0:= z_0\ge 0$. The paths of the volatility process~\eqref{eq:Z2} exhibit a rough behavior as is illustrated in Section~\ref{sec:simulation}. Note that \begin{equation}
 \lim_{\alpha\to -1} D^{\alpha+1}_0f(t) =f(t)
\end{equation}
which means that in the limiting case $\alpha\downarrow -1$ we obtain again the classical Heston model.
Now using the fact that for $\alpha\in (-1,-\frac12)$ we have
\begin{equation}\label{eq:Gamma2} (t-s)^{-\alpha-2} = \frac{\Gamma(-\alpha)}{\alpha+1}\int_0^\infty e^{-(t-s)x} \tilde{\mu}(dx),\quad\mbox{with  } \tilde{\mu}(dx)= \frac{x^{\alpha+1} dx}{ \Gamma(-\alpha)\Gamma(\alpha+1)}
\end{equation}
we obtain with the Fubini Theorem 
\begin{eqnarray}
\nonumber \nu_t &=& v_0 + Z_t  \frac{t^{-\alpha-1}}{\Gamma(-\alpha)} + \int_0^t (Z_t-Z_s)\int_0^\infty e^{-(t-s)x} \tilde{\mu}(dx)ds\\ \nonumber
&=& v_0 + Z_t  \frac{t^{-\alpha-1}}{\Gamma(-\alpha)} +\int_0^\infty  \int_0^t (Z_t-Z_s) e^{-(t-s)x} ds\tilde{\mu}(dx)\\
&=& v_0 + Z_t  \frac{t^{-\alpha-1}}{\Gamma(-\alpha)} +\int_0^\infty  \tilde{Y}_t^x \tilde{\mu}(dx)\label{eq:nuneg}
\end{eqnarray}
where 
\begin{equation}\tilde{Y}_t^x := \int_0^t (Z_t-Z_s) e^{-(t-s)x} ds.
\end{equation}
Using partial integration we see that $(\tilde{Y}_t^x)$ satisfies the stochastic differential equation
\begin{equation}
d\tilde{Y}_t^x = \Big(\frac{1}{x}(1-e^{-tx})\kappa(\theta-Z_t) - x \tilde{Y}_t^x\Big) dt +\frac{1}{x}(1-e^{-tx}) \sigma \sqrt{Z_t} dB_t^Z.
\end{equation}
Unfortunately it turns out that $\nu_t$ is not positive with probability one. To remedy this shortcoming we consider $a(\nu_t)$ as volatility with $a:\R\to \R_+$ sufficiently smooth, i.e. the stock price process $S=(S_t)$ is now given by
\begin{equation}\label{stockprice2_rough} d S_t = S_t \left( (r+\lambda a(\nu_t)) dt + \sqrt{a(\nu_t)}d B_t^S\right). \end{equation}
The  wealth process under an admissible portfolio strategy $\pi$ is thus given by the solution of the stochastic differential equation
\begin{eqnarray}
d W_t^\pi &=& W_{t}^\pi (r+ \pi_t \lambda a(\nu_t))dt + W_t^\pi\pi_t \sqrt{a(\nu_t)}d B_t^S,
\end{eqnarray} where we assume that ${W}_0=w_0>0$ is the given initial
wealth. We want to solve the same  optimization problem \eqref{eq:optprob} and proceed as in Section \ref{sec:approx}, i.e. we consider the  finite dimensional  approximation 
\begin{equation}\label{eq:nuapproxneg}
\nu_t ^n := v_0+Z_t  \frac{t^{-\alpha-1}}{\Gamma(-\alpha)}  +\int_0^\infty  \tilde{Y}_t^x  \tilde{\mu}^n(dx)=  v_0+Z_t  \frac{t^{-\alpha-1}}{\Gamma(-\alpha)}  + \sum_{i=1}^{n} \tilde{q}_i^n \tilde{Y}_t^{x_i^n}.
\end{equation}
The dynamics of the approximate stock price process are given by
\begin{equation}\label{stockprice2_} 
d S_t = S_t \left( (r+\lambda a(\nu_t^n)) dt + \sqrt{a(\nu_t^n)}d B_t^S\right)\end{equation}
and the stochastic differential equation  for the approximate wealth process is thus
\begin{equation}\label{eq:wealth2_}
d W_t^\pi =W_{t}^\pi \Big(r+ \pi_t \lambda a(\nu_t^n)\Big)dt + W_t^\pi \pi_t \sqrt{a(\nu_t^n)}d B_t^S.
\end{equation}

The finite dimensional classical stochastic optimal control problem is here defined by
\begin{equation}V(t,w,\tilde{y}_1,\ldots,\tilde{y}_n,z) := \sup_{\pi} \Eop_{t,w,\tilde{y}_1,\ldots,\tilde{y}_n,z}\left[\frac1\gamma\big({W}^{\pi}_T\big)^\gamma\right].
\end{equation}
where $\Eop_{t,w,\tilde{y}_1,\ldots,\tilde{y}_n,z}$ is the conditional expectation given $W_t=w, \tilde{Y}_t^{x_i}=\tilde{y}_i, Z_t=z$ at time $t$. As before portfolio strategies are $(\FC_t)$-adapted processes. In what follows we derive the corresponding Hamilton-Jacobi-Bellman (HJB) equation for this optimization problem. We denote the generic function by $G(t,w,\tilde{y}_1,\ldots,\tilde{y}_n,z)$ with $t\in[0,T], w> 0, \tilde{y}_i \ge 0, z> 0$. The boundary condition is given by $G(T,w,\tilde{y}_1,\ldots,\tilde{y}_n,z)= \frac 1\gamma w^\gamma$. In order to ease notation we set $\beta := a(v_0+z\frac{t^{-\alpha-1}}{\Gamma(-\alpha)}+\sum_{i=1} ^{n}  \tilde{q}_i \tilde{y}_i)$. Thus, the HJB equation reads
\begin{eqnarray}\label{eq:HJB1neg}
 \nonumber 0 = \sup_{u\in\R}\! \! \! \! && \! \! \! \! \Big\{G_t +G_w w(r+u\lambda \beta)   + \sum_{i=1}^n G_{\tilde{y}_i} \Big(\frac1{x_i}(1-e^{-tx_i})\kappa(\theta-z)-x_i\tilde{y}_i \Big) +G_z \kappa(\theta-z) \\ \nonumber
   &&+ \frac12 G_{ww}w^2 u^2 \beta  + \frac12 G_{zz} \sigma^2 z
   +\frac12 \sigma^2 z \sum_{i=1}^n\sum_{j=1}^n  G_{\tilde{y}_i\tilde{y}_j}\frac1{x_ix_j}  (1-e^{-tx_i})(1-e^{-tx_j})\\ \nonumber
   && + \sigma^2 z  \sum_{i=1}^n G_{\tilde{y}_iz} \frac1{x_i}(1-e^{-tx_i}) +G_{wz}wu\sigma\rho\sqrt{z\beta}\\
   &&+\sum_{i=1}^nG_{w\tilde{y}_i}\rho wu\sigma \sqrt{z\beta}\frac1{x_i}(1-e^{-tx_i})\Big\}.
\end{eqnarray}

To solve the Hamilton-Jacobi-Bellman equation, we start with the same transformation as in the proof of Theorem \ref{theo:HJB}. Using the Ansatz $G(t,w,\tilde{y}_1,\ldots,\tilde{y}_n,z)= \frac1\gamma w^\gamma f(t,\tilde{y}_1,\ldots,\tilde{y}_n,z)$ with $f(T,\tilde{y}_1,\ldots,\tilde{y}_n,z)=1$ and  plugging it into \eqref{eq:HJB1neg} where we use the abbreviation  $h_i(t) = \frac{1}{x_i}(1-e^{-tx_i})$ we obtain:
\begin{eqnarray}
\nonumber  0 = \sup_{u\in\R} \! \! \! \! &&\!\! \! \!  \Big\{f_t +f(r+u\lambda \beta)\gamma   + \sum_{i=1}^n f_{y_i} \big( h_i(t)\kappa(\theta-z)-x_i\tilde{y}_i\big) + f_z\kappa(\theta-z)\\ \nonumber
   && +\frac12 (\gamma-1)\gamma u^2 \beta f + \frac12 f_{zz} \sigma^2 z+\frac12 \sigma^2 z\sum_{i=1}^n  \sum_{j=1}^n f_{\tilde{y}_i\tilde{y}_j}h_i(t) h_j(t) \\
   && + \sigma^2 z \sum_{i=1}^n  f_{z\tilde{y}_i}h_i(t) +f_z \gamma u\sigma\rho \sqrt{z\beta}+ \sum_{i=1}^n f_{\tilde{y}_i} \gamma u\sigma\rho\sqrt{z\beta} h_i(t)\Big\}.
\end{eqnarray}
Maximizing this expression in $u$ gives \begin{equation}
u^*_t = \frac{\lambda}{1-\gamma} + \frac{\sigma\rho}{1-\gamma}\sqrt{\frac{z}{\beta}} \Big( \frac{f_z}{f}+\frac{\sum_{i=1}^n f_{\tilde{y}_i}h_i(t)}{f}\Big). 
\end{equation} In case $\rho=0$ we get again $u^*_t = \frac{\lambda}{1-\gamma} $ independent of $t,w,\tilde{y}_1,\ldots,\tilde{y}_n$. Inserting the maximum point yields
\begin{eqnarray}\label{eq:pdefneg}
 \nonumber\!\! 0 \!\!&=&\!\!\! f_t +\frac12 f\gamma \frac{\lambda^2 \beta}{1-\gamma} +f\gamma r \!+ \!\sum_{i=1}^n f_{y_i} \big( h_i(t) \kappa(\theta\! -\! z)-x_i\tilde{y}_i\big) \! +\! f_z\kappa (\theta\!-\! z) + \!\!\frac12 \sigma^2 z f_{zz}\\ \nonumber
  && +\frac12 \sigma^2 z\sum_{i=1}^n  \sum_{j=1}^n f_{\tilde{y}_i\tilde{y}_j}h_i(t) h_j(t) + \sigma^2 z \sum_{i=1}^n  f_{z\tilde{y}_i}h_i(t)\\ 
  && + \frac12 \frac{z}{f} \frac{\gamma\sigma^2\rho^2}{1-\gamma} \Big(f_z+ \sum_{i=1}^n f_{\tilde{y}_i} h_i(t)\Big)^2+\frac{\sigma\rho\lambda\gamma}{1-\gamma}\sqrt{z\beta} \Big(f_z+ \sum_{i=1}^n f_{\tilde{y}_i} h_i(t)\Big).
\end{eqnarray}
Unfortunately this PDE is rather involved and has to be solved numerically. In case $\rho=0$ we obtain:

\begin{theorem}\label{theo:FKneg}
Suppose $\rho=0$ and a classical solution $f$ of the partial differential equation \eqref{eq:pdefneg} with boundary condition $f(T,\tilde{y}_1,\ldots,\tilde{y}_n,z)=1$ exists for $t\in[0,T], T\le T_\infty, w>0, \tilde{y}_i>0, z>0$. Then it can  be written as
$$ f(t,\tilde{y}_1,\ldots,\tilde{y}_n,z) = \Eop_{t,\tilde{y}_1,\ldots,\tilde{y}_n,z}\left[ e^{\int_t^T \Big( \gamma r+ \frac{1}{2}\frac{\gamma \lambda^2}{1-\gamma} a(\nu_s^n)  \Big) ds} \right]$$
where $(\nu_t^n)$ is given by \eqref{eq:nuapproxneg}.
\end{theorem}

The proof follows like the proof of Theorem \ref{theo:FK} with the Feynman-Kac theorem given in \cite{HS00}.

\begin{remark}
It can be seen from the simulation results in Section \ref{sec:simulation} that if we choose the model parameters in the right way, the paths of $(\nu_t)$ stay positive with very high probability. In particular note that for $\alpha \downarrow -1$ we obtain in the limit the classical Heston model where paths are positive with probability one. Moreover, on the domain $D := \{(t,\tilde{y}_1,\ldots,\tilde{y}_n,z) :  v_0+z\frac{t^{-\alpha-1}}{\Gamma(-\alpha)}+\sum_{i=1} ^{n}  \tilde{q}_i \tilde{y}_i > 0 \}$  we get a classical solution of \eqref{eq:pdefneg}. This is shown in the appendix.
\end{remark}

The verification that the HJB equation indeed yields the value function works  the same way as in the case $\alpha>0$. We obtain:

\begin{theorem}\label{theo:valuefiniteneg}[Verification]
Suppose a solution $f$ of the partial differential equation \eqref{eq:pdefneg} in case $\rho=0$ with boundary condition $f(T,\tilde{y}_1,\ldots,\tilde{y}_n,z)=1$  exists. For $t\in[0,T], T\le T_\infty$, the optimal investment strategy $(\pi_t^*)$ is given by $$\pi_t^*\equiv \frac{\lambda}{1-\gamma}\,,$$ and the value function can be written as
\begin{equation}\label{eq:Vexpressionneg}
V(t,w,\tilde{y}_1,\ldots,\tilde{y}_n,z) = \frac1\gamma w^\gamma  f(t,\tilde{y}_1,\ldots,\tilde{y}_n,z).
\end{equation}
\end{theorem}

\begin{proof}
As mentioned before, the proof is essentially the same as for Theorem \ref{theo:valuefinite}. We only have to replace $G(t,w,\tilde{y}_1,\ldots,\tilde{y}_n,z) $ by $\frac1\gamma w^\gamma  f(t,\tilde{y}_1,\ldots,\tilde{y}_n,z)$ and $\nu_t$ by $a(\nu_t)$.
\end{proof}

The final step now is to take the limit $n\to\infty$ and consider the optimization problem without approximation. Here we obtain:

\begin{theorem}\label{theo:limitcaseneg} [Solution for the rough path Problem]
Suppose a solution $f$ of the partial differential equation \eqref{eq:pdefneg} in case $\rho=0$ with boundary condition $f(T,\tilde{y}_1,\ldots,\tilde{y}_n,z)=1$  exists. The optimal portfolio strategy for problem \eqref{eq:optprob} with $\alpha\in (-1,-\frac12)$  is given by $\pi_t^*\equiv \frac{\lambda}{1-\gamma}$  and the value function can be written as
\begin{equation}\label{eq:valueroughneg}
V(w_0,v_0,z_0) = \frac1\gamma w_0^\gamma \;\Eop_{0,\tilde{y}_1,\ldots,\tilde{y}_n,z_0}\left[ e^{\int_0^T \Big( \gamma r+ \frac{1}{2}\frac{\gamma \lambda^2}{1-\gamma} a(\nu_s)  \Big) ds} \right]\end{equation}
where $(\nu_t)$ is given by \eqref{eq:nuneg}. In case $\gamma<0$ we have to assume that $(W_t^{\pi,n})$ is uniformly integrable for all $\pi$.
\end{theorem}

\section{Simulation Results}\label{sec:simulation}

To illustrate the fractional and rough behavior of the processes~\eqref{vola} and~\eqref{vola2}, we have derived and implemented the corresponding explicit forward Euler methods, which are inspired by the survey~\cite{G18}; the resulting schemes are for a constant step-size $h > 0$: 
\begin{align}
\nu_k &= \nu_0 + h^{\alpha} \sum_{j=0}^{k-1} \left(  \frac{(k-j)^{\alpha} - (k-j-1)^{\alpha}}{\Gamma(\alpha+1)} \,Z_j  \right)\,,\\
\nu_k &= \nu_0 +  \frac{Z_k\,k^{-\alpha-1}}{\Gamma(-\alpha)} \nonumber\\
&+ \frac{1}{\Gamma(-\alpha)\,h^{\alpha+1}} \frac{(\alpha+1)}{(\alpha+0.5)} \sum_{j=0}^{k-1} \left(  \frac{Z_k-Z_j}{(k-j)^{\delta}} \left\{ \frac{1}{(k-j-1)^{\alpha+1-\delta}} - \frac{1}{(k-j)^{\alpha+1-\delta}} \right\}  \right)\,,
\end{align}
where we choose $\delta$ close to, but smaller than $0.5$. In the analysis to follow we have chosen the step size as $h = 0.001$ if not stated otherwise.

\subsection{The Fractional Volatility Case}

Figure~\ref{fig:vola_frac} shows five sample paths of the Cox-Ingersoll-Ross process~\eqref{eq:Z} with its corresponding stock price process together with the fractional volatility process \eqref{vola} and the fractional stock price process \eqref{stockprice} for \[ \alpha \in \{0.05,0.5,0.95\} \mbox{ and } \rho \in \{-0.7,0,0.7 \}\,,\] where the remaining model parameters have been chosen as
\begin{align}
 T = 1\,,S_0 = 100\,,r=0.02\,, \lambda = 0.5\,, \theta = 0.05\,,\kappa = 6\,,v_0 = 0\,,z_0=0.05\,. \label{parameters_1}
\end{align}  

We deduce from Figure~\ref{fig:vola_frac} that also visually the convergence to the classical Heston model in the limiting case $\alpha\downarrow 0$ holds true. We observe the increasing smoothness of the fractional paths for higher values of $\alpha$. Moreover, the long-range dependence of the fractional volatility process as indicated by~\eqref{eq:long_range_1} and \eqref{eq:long_range_2} is clearly detectable, which is also in line with the increasing behavior of the value function~\eqref{eq:value_frac} in the fractional model with an increasing history. 

\subsection{The Rough Volatility Case}

Figure~\ref{fig:vola_rough} shows five sample paths of the Cox-Ingersoll-Ross process~\eqref{eq:Z} with its corresponding stock price process together with the rough volatility process \eqref{vola2} and the rough stock price process \eqref{stockprice2_} (with $a(\nu_t)=|\nu_t|$) for \[ \alpha \in \{-0.95,-0.75,-0.55\} \mbox{ and } \rho \in \{-0.7,0,0.7 \} \,,\] where the remaining model parameters have been chosen as in \eqref{parameters_1}.

\begin{landscape}
\begin{figure}[htb]
    \vspace{1.5cm}
	\includegraphics[width=1.5\textwidth]{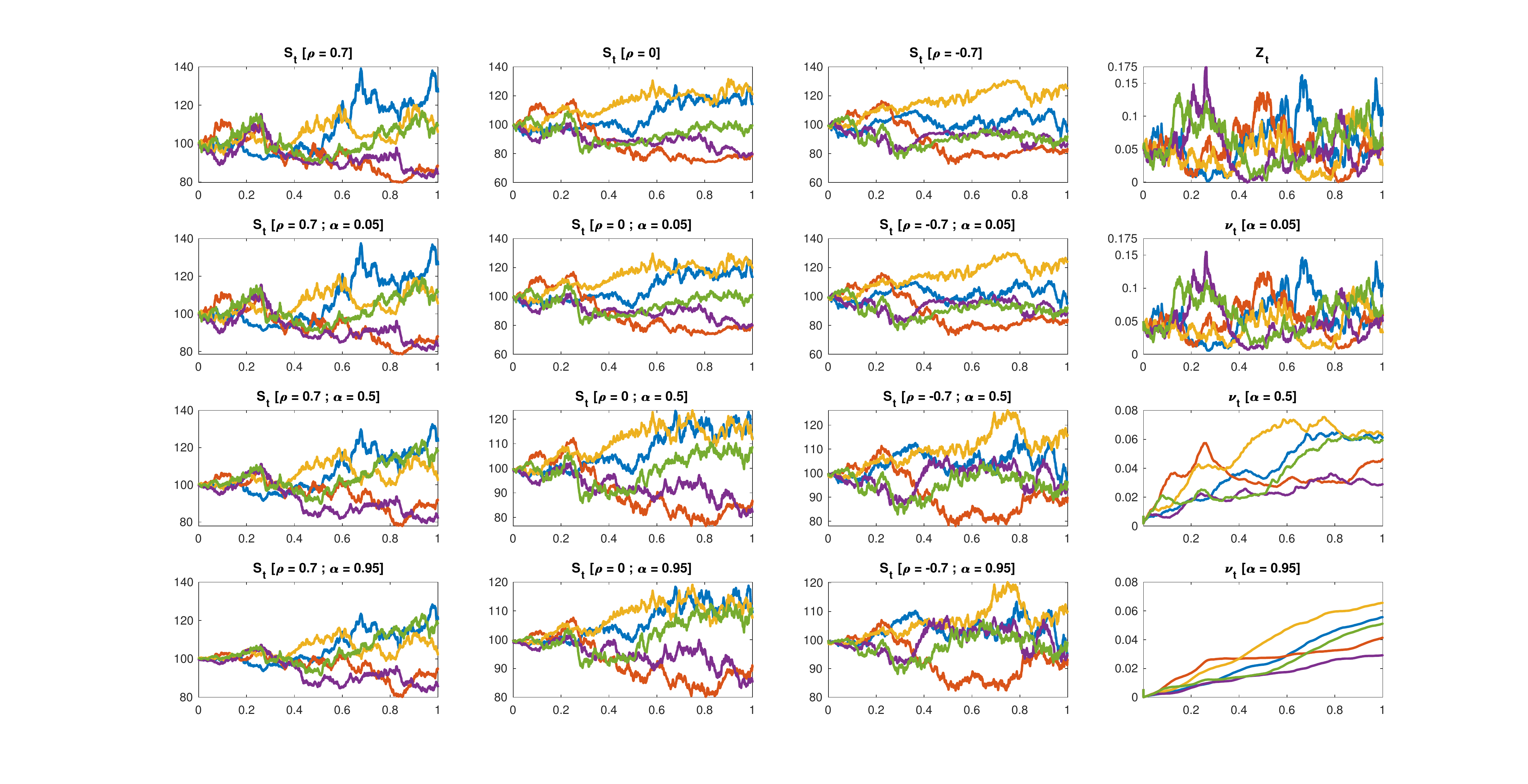}
	\vspace{-1.5cm}
	\caption{Fractional volatility and stock price paths for $\alpha\in\{0.05,0.5,0.95\}$ in comparison with the classical Heston model.}\label{fig:vola_frac}
\end{figure}
\end{landscape}

\begin{landscape}
\begin{figure}[htb]
    \vspace{1.5cm}
	\includegraphics[width=1.5\textwidth]{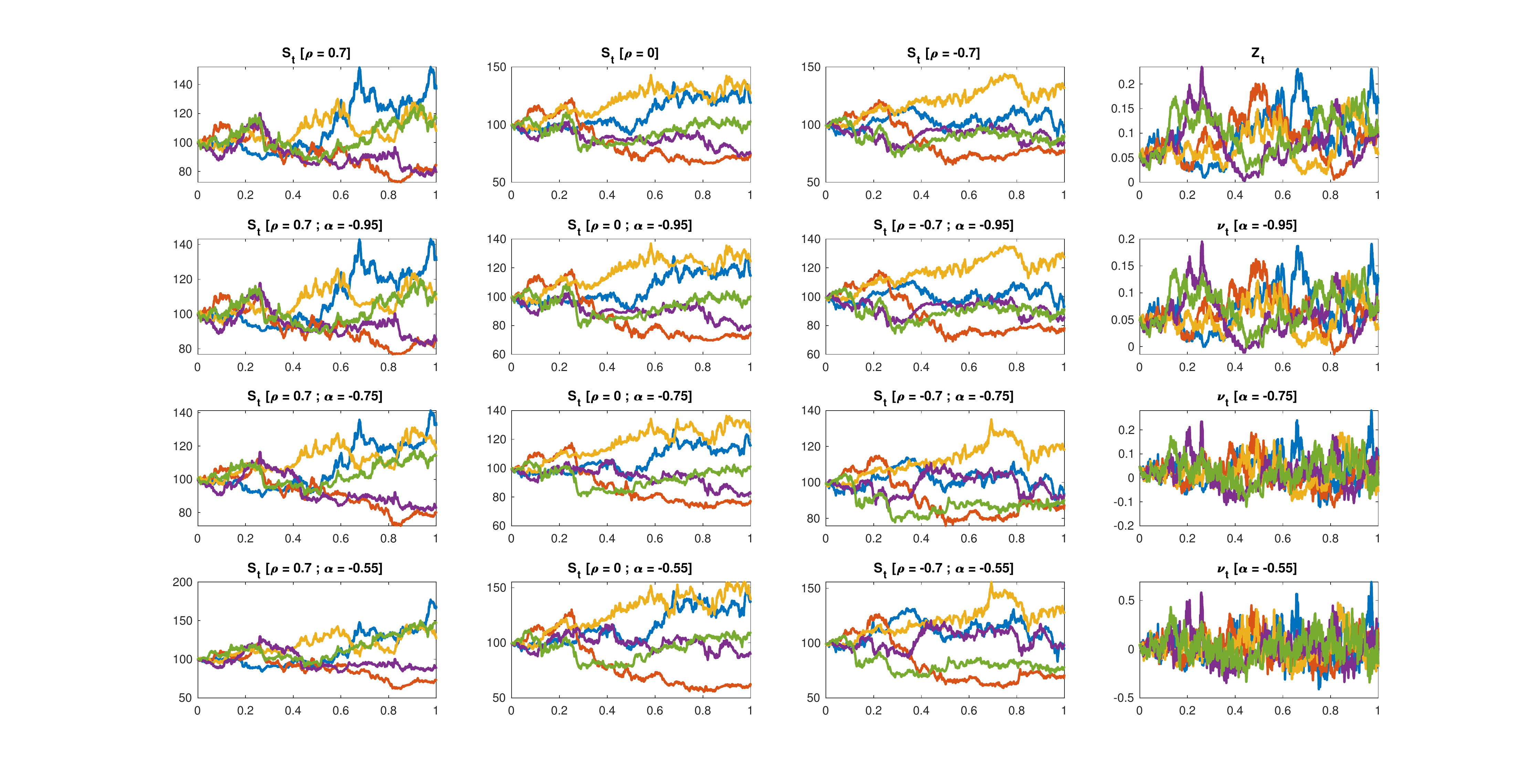}
	\vspace{-1.5cm}
	\caption{Rough volatility and stock price paths for $\alpha\in\{-0.95,-0.75,-0.55\}$ in comparison with the classical Heston model.}\label{fig:vola_rough}
\end{figure}
\end{landscape}

We deduce from Figure~\ref{fig:vola_rough} that also in the rough case the expected convergence to the classical Heston model in the limiting case $\alpha\downarrow -1$ holds true. We  observe the increasing roughness of the paths for values of $\alpha$ closer to the maximal possible value of $-\tfrac{1}{2}$. Note that in this case the roughness of the paths has a significant impact on the stock price. We also detect the effect that the rough volatility process does not stay positive with probability one as discussed in Section~\ref{sec:rough}, and this effect is more pronounced, the rougher the paths are. We wish to stress that this behavior is highly dependent on the choice of the starting value $v_0$ of the rough CIR process, which we have chosen as $v_0=0$, in order to obtain the convergence  $ \underset{\alpha\to -1}{\lim}(\nu_t) =Z_t$ as indicated by Figure~\ref{fig:vola_rough}. In that regard, Figure~\ref{fig:vola_rough_3} shows that for instance in the case $\alpha=-0.75$ the rough volatility process remains strictly positive for the choice $v_0=3, z_0=0.15$.

\begin{figure}[htb]
\includegraphics[width=0.75\textwidth]{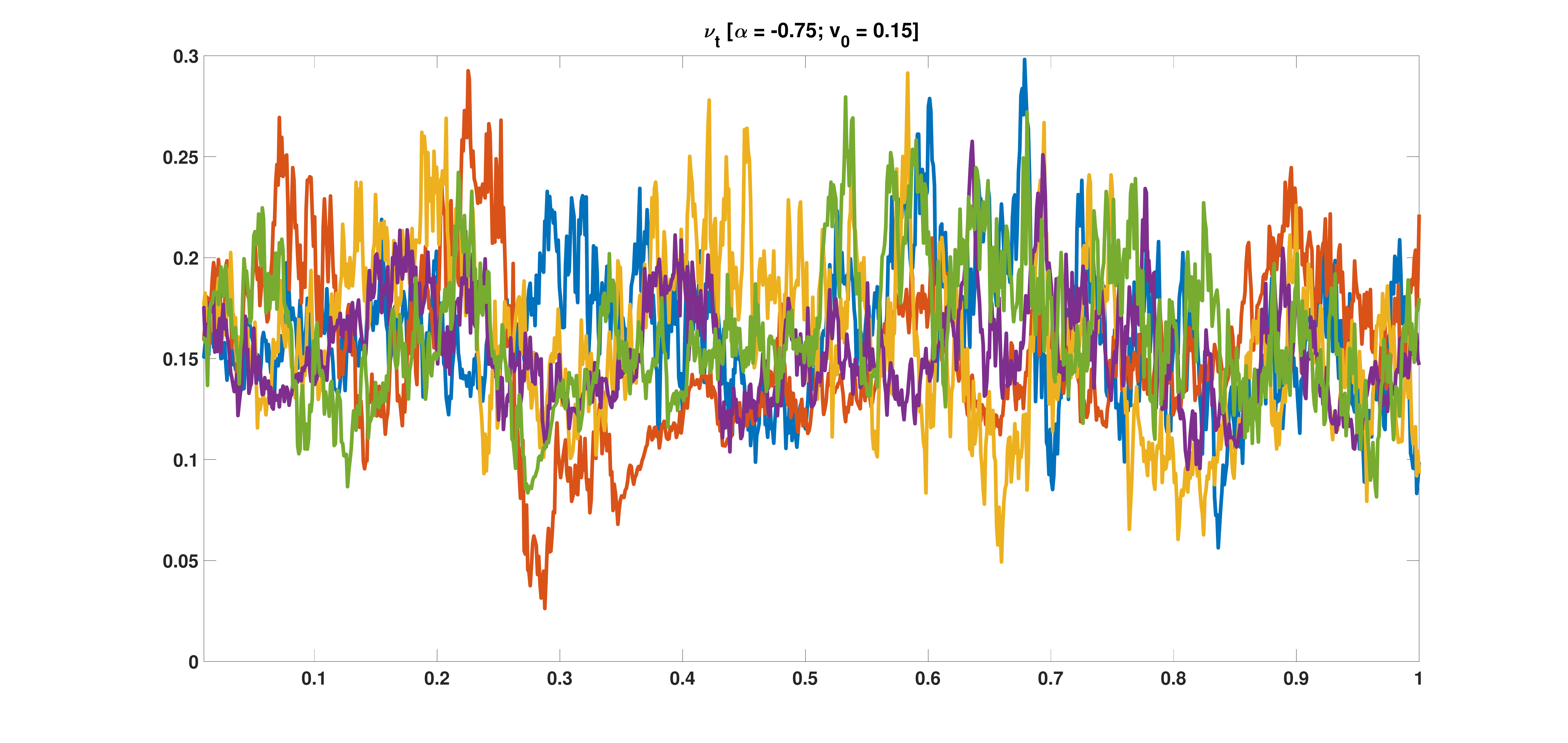}
\caption{Rough volatility paths for $\alpha=-0.75$ and $v_0=0.15$.}\label{fig:vola_rough_3}
\end{figure}

Figure~\ref{fig:vola_rough_exp} further depicts the choices of the absolute value and the exponential function for the function $a(\nu_t)$  from $a:\R\to \R_+$ with different starting values $v_0$. As a result, we deduce that both choices yield the desired effect. We note for the sake of completeness that we plugged in the absolute value in the equation for the stock price for the illustrations in Figure~\ref{fig:vola_rough}. 

\begin{figure}[htb]
\includegraphics[width=\textwidth]{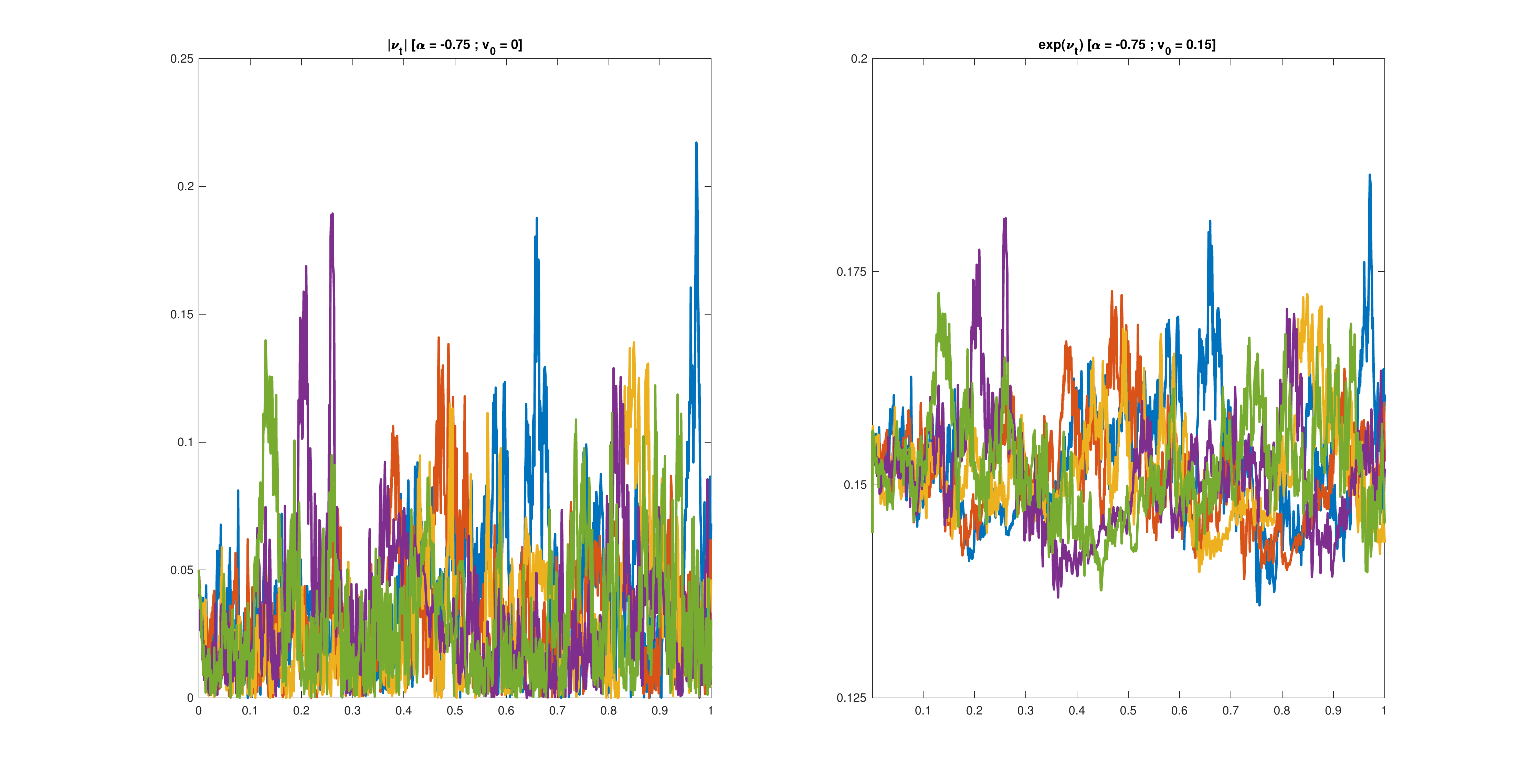}
\caption{Absolute and exponential rough volatility paths for $\alpha=-0.75$.}\label{fig:vola_rough_exp}
\end{figure}

\subsection{Optimal Terminal Wealth}

To illustrate the impact of the fractional and rough volatility process on the optimal terminal wealth in the case $\rho = 0$, Figure~\ref{fig:frac_wealth} depicts the optimal wealth process
\begin{align}\label{eq:wealth_opt}
W_t^{\pi^\star} = w_0 \exp\left( rt + \int_0^t \left( \frac{\lambda^2}{1-\gamma} \nu_s - \frac{\lambda^2}{(1-\gamma)^2} \nu_s \right) ds + \int_0^t\frac{\lambda}{1-\gamma} \sqrt{\nu_s} dB^S_s \right)\,,
\end{align}
which we have obtained by solving the SDE~\eqref{eq:wealth_process} explicitly and plugging in the optimal (Merton) portfolio strategy, and where we have chosen the additional model parameters as 
\begin{align}\label{parameters_2}
 w_0 = 1000\,,\,\gamma=-2\,.
\end{align} 
We observe that the roughness of the volatility paths (RHS) causes a higher variance of the paths of the optimal wealth as in the smooth case (LHS), increasing the closer the values of $\alpha$ are to $-\tfrac{1}{2}$.

\begin{figure}[htb]
\includegraphics[width=0.49\textwidth]{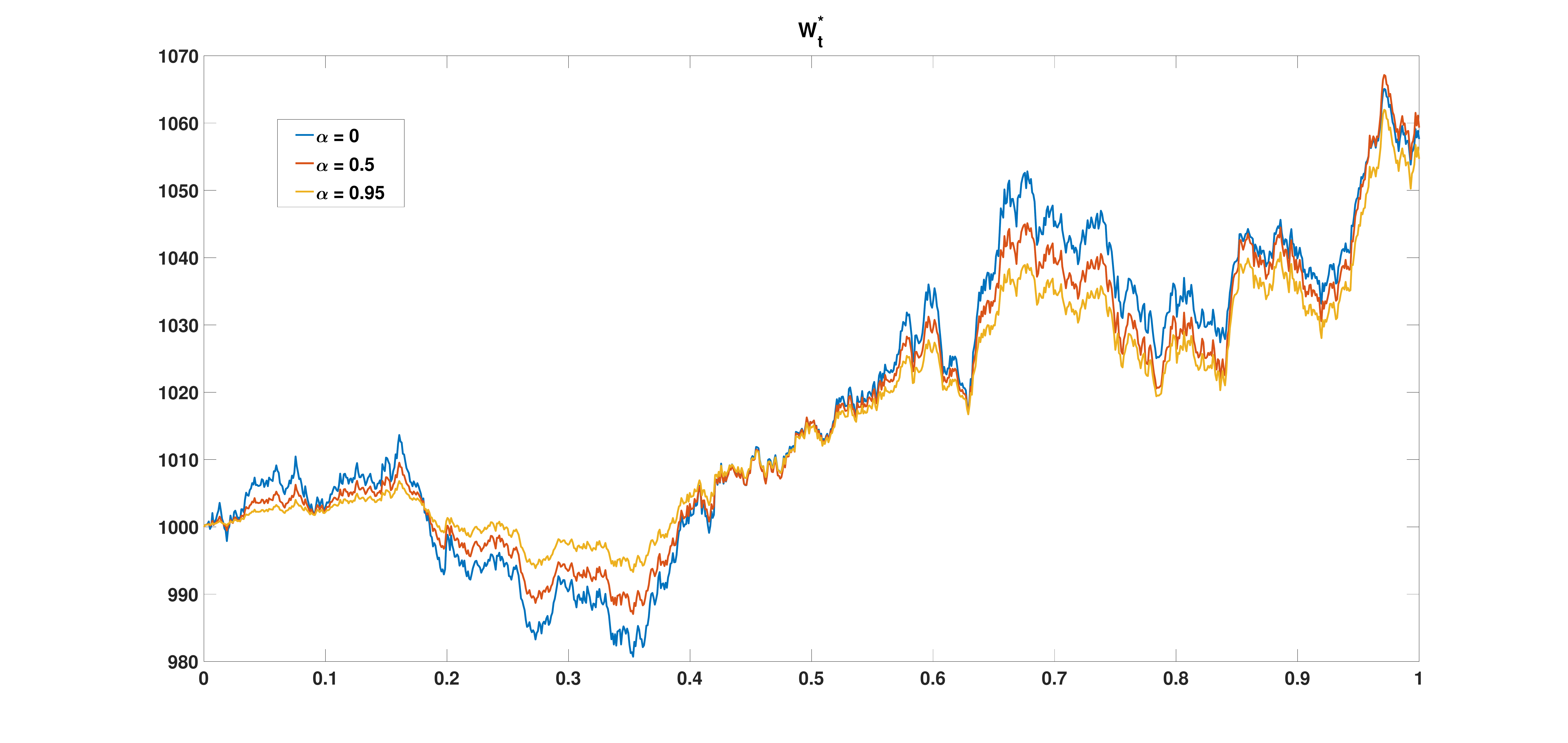}
\includegraphics[width=0.49\textwidth]{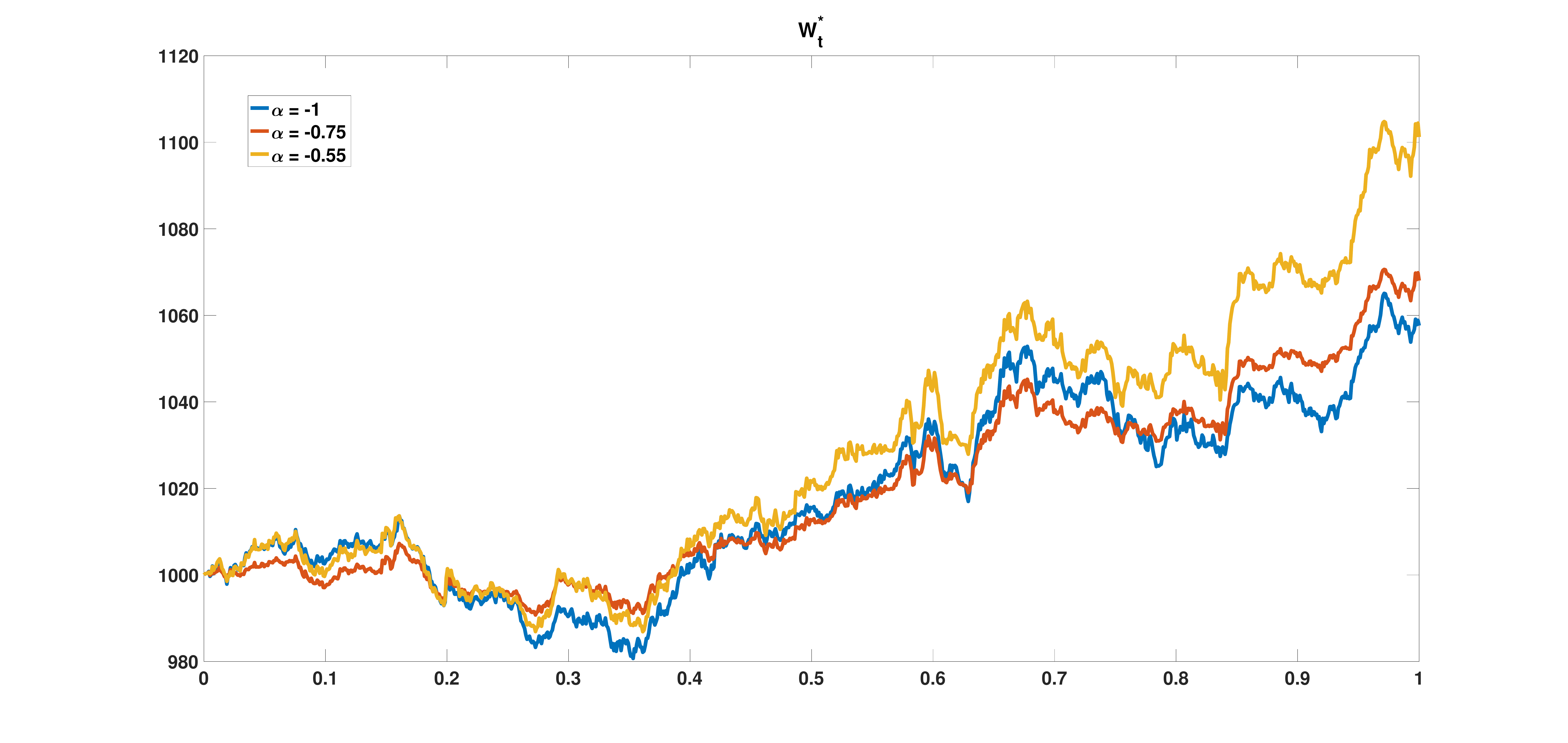}
\caption{Fractional optimal terminal wealth for $\alpha\in\{0.5,0.95\}$ (LHS) and rough optimal terminal wealth for $\alpha\in\{-0.75,-0.55\}$ (RHS) compared to the classical Heston model ($\alpha = 0$ respectively $\alpha = -1$).}\label{fig:frac_wealth} 
\end{figure}

\subsection{Long Term Behavior} To conclude our simulation study, we compare the long term behavior of the fractional ($\alpha = 0.75$) and the rough ($\alpha = -0.75$) volatility processes, together with the corresponding stock price processes in Figure~\ref{fig:long_term}. Again, in line with the increasing behavior of the value function~\eqref{eq:value_frac} in the fractional model, we observe that in the long term, the fractional stock price process reaches unrealistically high values across different correlation levels. This is caused by the upwards trend of the fractional volatility process as is reflected by the sample paths in Figure~\ref{fig:long_term}.  In contrast, the rough stock price process, which is driven by the absolute value of the rough volatility process, moves within a reasonable range. We thus recommend to use the fractional model only for short term investment horizons, whereas the rough model seems as well suitable for long time investment horizons.

\begin{figure}[htb]
\includegraphics[width=\textwidth]{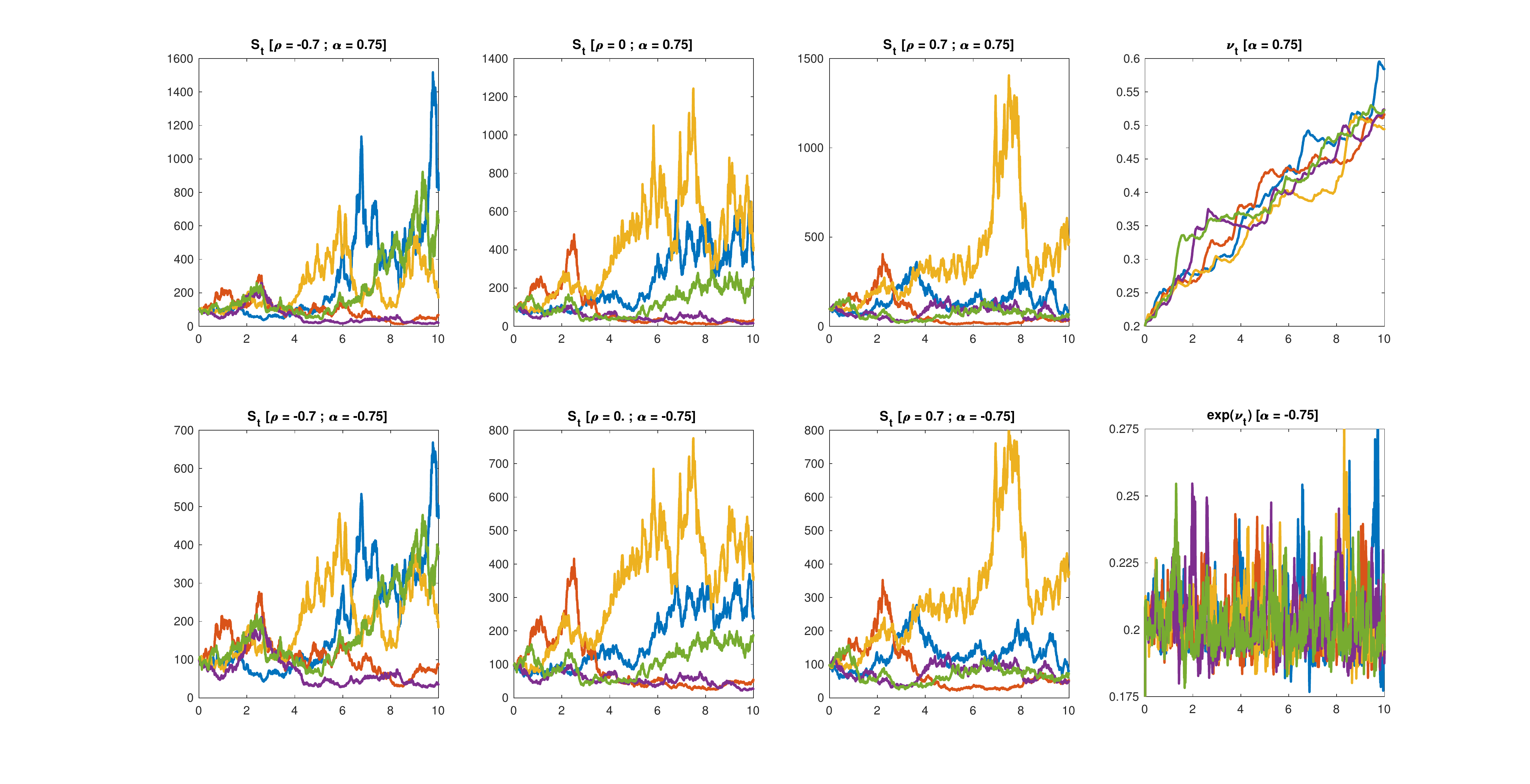}
\caption{Comparison of the long term behavior of fractional and rough volatility processes with corresponding stock price processes.}\label{fig:long_term}
\end{figure}

\section{Appendix}

\subsection{H\"older-continuity of the Cox-Ingersoll-Ross model}
The following lemma may be common knowledge but since we have not found it somewhere we will provide a proof. We credit main ideas to \cite{TR}.

\begin{lemma}
The paths of $(Z_t)$ given in \eqref{vola} are almost $\frac12$- H\"older-continuous.
\end{lemma}

\begin{proof}
We use Kolmogorov's continuity theorem which states that for a stochastic process $(X_t)$ which satisfies
$$ \Eop\big[ | X_t-X_s|^\alpha\big] \le K |t-s|^{1+\beta}$$
for positive constants $K,\alpha,\beta$ there exists a modification of $(X_t)$ with paths which are $\gamma$-H\"older continuous for  $0<\gamma<\frac{\beta}{\alpha}$. In view of \eqref{vola} we only have to deal with the stochastic integral $\int_0^t \sqrt{Z_u} dB_u^Z$ since the remaining part is continuous. In what follows let $p>2$ and consider
\begin{eqnarray*}
&& \Eop\Big[\big| \int_0^t \sqrt{Z_u} dB_u^Z -\int_0^s\sqrt{Z_u} dB_u^Z\big|^{2p}\Big] =   \Eop\Big[\big| \int_s^t \sqrt{Z_u} dB_u^Z \big|^{2p}\Big]\\
&\le& K_p \Eop\Big[\big( \int_s^t Z_u du  \big)^{p}\Big] \le K_p (t-s)^{p-1} \Eop\Big[\big(\int_s^t {Z_u}^p du\big)\Big] = K_p (t-s)^{p-1} \int_s^t \Eop[{Z_u}^p] du,
\end{eqnarray*}
where we have used the Burkholder-Davis-Gundy inequality for the first inequality and the H\"older inequality for the second one. It is well-known that $Z_t$ has a non-central Chi-squared distribution thus the last integral is finite. Hence we can apply the Kolmogorov continuity theorem and obtain that $(Z_t)$ has a modification with paths which are $\gamma$-H\"older continuous for  $0<\gamma<\frac12 -\frac{1}{p}$.
\end{proof}

\subsection{Classical solution of \eqref{eq:pdefneg}}
Note that the initial value of $\nu_0$ should be chosen positive such that we start within the domain $D := \{(t,\tilde{y}_1,\ldots,\tilde{y}_n,z) :  v_0+z\frac{t^{-\alpha-1}}{\Gamma(-\alpha)}+\sum_{i=1} ^{n}  \tilde{q}_i \tilde{y}_i > 0 \}$. As long as the process $(\nu_t)$ stays positive it is in $D$ and on this domain we can skip the  absolute value and obtain the partial differential equation
\begin{eqnarray}\label{eq:pdefnegwithout}
 \nonumber\!\! 0 \!\!&=&\!\!\! f_t +\frac12 f\gamma \frac{\lambda^2 (v_0+z\frac{t^{-\alpha-1}}{\Gamma(-\alpha)}+\sum_{i=1} ^{n}  \tilde{q}_i \tilde{y}_i)}{1-\gamma} +f\gamma r \!+ \!\sum_{i=1}^n f_{y_i} \big( h_i(t) \kappa(\theta\! -\! z)-x_i\tilde{y}_i\big) \! +\! f_z\kappa (\theta\!-\! z) \\
  && + \!\!\frac12 \sigma^2 z f_{zz}+\frac12 \sigma^2 z\sum_{i=1}^n  \sum_{j=1}^n f_{\tilde{y}_i\tilde{y}_j}h_i(t) h_j(t) + \sigma^2 z \sum_{i=1}^n  f_{z\tilde{y}_i}h_i(t)
\end{eqnarray}
with boundary condition $f(T,\tilde{y}_1,\ldots,\tilde{y}_n,z)=1$.
Now using again the Ansatz
$$ f(t,\tilde{y}_1,\ldots,\tilde{y}_n,z)= \exp\Big(\phi(T-t)+\sum_{i=1}^n \psi_i(T-t)\tilde{y}_i+\varphi(T-t)z \Big)  $$
with $\phi(0)=\psi_i(0)=\varphi(0)=0$ we obtain a solution of \eqref{eq:pdefneg} with
\begin{equation}\label{eq:psiineg}
\psi_i(T-t) =  \eta \tilde{q}_i \int_0^{T-t}e^{ - x_is} ds
\end{equation}
where again $\eta := \frac12\frac{\gamma\lambda^2}{1-\gamma}$ and $\varphi$ and $\phi$ are solutions of the ordinary differential equations
\begin{eqnarray}
 \label{riccati1neg}
  \varphi_t(T-t) &=& \! \! \frac{\eta t^{-\alpha-1}}{\Gamma(-\alpha)}\!
    \! -\! \kappa\varphi(T\!\!-\!\! t)+ \frac12\sigma^2 \varphi^2(T\!\! -\!\! t)\! -\! \eta h^n(t)  
     \big(\kappa\! -\! \sigma^2\varphi(T\!\! -\!\! t)\! -\! \frac12 \sigma^2 \eta  h^n(t)   \big)    
   \\ \label{ode2neg}
    \phi_t(T-t) &=& \gamma r+v_0\eta+\theta\kappa\Big( \varphi (T-t)+ h^n(t)  \Big).
\end{eqnarray}
with $h^n(t):=\int_0^t\int_0^{T-t} \int_0^\infty e^{-x(s+u)} \tilde{\mu}^{n}(dx)duds  $ and boundary condition $\varphi(0)=\phi(0)=0$. Note that both ordinary differential equations \eqref{riccati1neg} and \eqref{ode2neg} have solutions due to the Picard-Lindel\"of Theorem. Thus in total we get a classical solution of \eqref{eq:pdefneg} on $D$.

\subsection{Additional Proofs}
This part of the appendix comprises longer proofs and auxiliary lemmas.

{\em Proof of Lemma \ref{lem:conv1}:}
For $n\in\N$ define the function
$$t^{n}(x) := \sum_{i=0}^{n-1} x_{i+1}^{n}1_{[\xi_{i}^{n},\xi_{i+1}^{n}]}(x), \quad x\ge 0.$$
Obviously it holds that $t^{n}(x) \to x$ for $n\to\infty$ under assumptions (i)-(iii). Thus we have 
\begin{eqnarray*}
\int fd \mu^{n}  &=& \sum_{i=1}^n q_i^{n} f(x_i^{n}) \\
&=& \sum_{i=0}^{n-1} \int_{\xi_i^{n}}^{\xi_{i+1}^{n}}\mu(dx) f(x_{i+1}^{n}) = \int_{\xi_0^{n}}^{\xi_{n}^{n}} f(t^{n}(x))\mu(dx)
\end{eqnarray*}
which implies the first convergence statement since $f$ is continuous and thus also bounded on compact intervals. Now suppose that  $ \mathcal{Z}^{n+1}$ differs from  $\mathcal{Z}^{n}$ by only one point. We show that $\int fd \mu^{n}  \le \int fd \mu^{n+1}$.  If this point is added on $(0,\xi_0^{n})$ or on $(\xi_n^{n},\infty)$ the statement follows immediately since $f\ge 0$. Now suppose a point $\xi$  is added on $(\xi_i^{n},\xi_{i+1}^{n})$. Let us denote by $x_{i+1}^{n+1}$ and $x_{i+2}^{n+1}$ the two new points in $\mathcal{Z}^{n+1}$ which replace $x_{i+1}^{n}$. We obtain:
\begin{eqnarray*}
x_{i+1}^{n} &=& \frac{\int_{\xi_i^{n}}^{\xi_{i+1}^{n}}x\mu(dx)}{\int_{\xi_i^{n}}^{\xi_{i+1}^{n}}\mu(dx)} = \frac{\int_{\xi_i^{n}}^{\xi}x\mu(dx)}{\int_{\xi_i^{n}}^{\xi}\mu(dx)} \cdot \frac{\int_{\xi_i^{n}}^{\xi}\mu(dx)}{\int_{\xi_i^{n}}^{\xi_{i+1}^{n}}\mu(dx)}+ \frac{\int_{\xi}^{\xi_{i+1}^{n}}x\mu(dx)}{\int_{\xi}^{\xi_{i+1}^{n}}\mu(dx)}
\cdot \frac{\int_{\xi}^{\xi_{i+1}^{n}}\mu(dx)}{\int_{\xi_i^{n}}^{\xi_{i+1}^{n}}\mu(dx)}.
\end{eqnarray*}
Thus, with the convexity of $f$ it follows that
$$ f(x_{i+1}^{n} ) \le f(x_{i+1}^{n+1} ) \frac{\int_{\xi_i^{n}}^{\xi}\mu(dx)}{\int_{\xi_i^{n}}^{\xi_{i+1}^{n}}\mu(dx)} + f(x_{i+2}^{n+1} )\frac{\int_{\xi}^{\xi_{i+1}^{n}}\mu(dx)}{\int_{\xi_i^{n}}^{\xi_{i+1}^{n}}\mu(dx)}$$
Multiplying both sides with $\int_{\xi_i^{n}}^{\xi_{i+1}^{n}}\mu(dx)$ implies
that replacing the term $\int_{\xi_i^{n}}^{\xi_{i+1}^{n}}\mu(dx) f(x_{i+1}^{n} )$ by $\int_{\xi_i^{n}}^{\xi}\mu(dx)f(x_{i+1}^{n+1} )+ \int_{\xi}^{\xi_{i+1}^{n}}\mu(dx)f(x_{i+2}^{n+1} )$ increases the value. Thus monotone convergence follows. \hfill $\square$

{\em Proof of Theorem \ref{theo:HJB}:} 
In order to simplify this Hamilton-Jacobi-Bellman equation, we choose the usual separation Ansatz with $G(t,w,y_1,\ldots,y_n,z)= \frac1\gamma w^\gamma f(t,y_1,\ldots,y_n,z)$ and boundary condition given by $f(T,y_1,\ldots,y_n,z)=1$.  Plugging this Ansatz into \eqref{eq:HJB1} We obtain:
\begin{eqnarray*}
\nonumber  0 = \sup_{u\in\R} \! \! \! \! &&\! \! \! \!   \Big\{f_t +f(r+u\lambda \beta)\gamma  + \sum_{i=1}^n f_{y_i} (z-x_iy_i) + f_z\kappa(\theta-z)\\
   && +\frac12 (\gamma-1)\gamma u^2 \beta f + \frac12 f_{zz} \sigma^2 z+f_z \gamma\sigma u\rho \sqrt{z\beta}\Big\}.
\end{eqnarray*}
Maximizing this expression in $u$ gives 
$$u^*(y_1,\ldots,y_n,z) = \frac{\lambda}{1-\gamma}+\frac{\sigma\rho}{1-\gamma}\sqrt{\frac{z}{\beta}} \frac{f_z}{f}.$$ Note that for $\rho=0$ this reduces to $u^*= \frac{\lambda}{1-\gamma}$ independent of $t,w,y_1,\ldots,y_n$. Inserting the maximum point yields
\begin{eqnarray}\label{eq:pdef1}
\nonumber  0 &=& f_t +\frac12 f\gamma \frac{\lambda^2 \beta}{1-\gamma} +f\gamma r + \sum_{i=1}^n f_{y_i} (z-x_iy_i) +f_z \Big( \kappa (\theta-z)+\frac{\lambda\gamma\sigma\rho \sqrt{z\beta}}{1-\gamma}\Big) \\
  && + \frac12 \sigma^2 z f_{zz}+ \frac12 \frac{\gamma\sigma^2\rho^2 z}{1-\gamma} \frac{f_z^2}{f}.
\end{eqnarray}

In order to further simplify the equation  we use the Ansatz 
\begin{equation}\label{eq:f2}
f(t,y_1,\ldots,y_n,z) = g(t,y_1,\ldots,y_n,z)^c
\end{equation}
with $c= \frac{1-\gamma}{1-\gamma+\gamma\rho^2}$ and $g(T,y_1,\ldots,y_n,z)=1$. For similar transformations see e.g. \cite{Z01,K05,BL13}. Inserting the derivatives and rearranging the terms leads to
\begin{equation*}
0= cg_t +g\Big(\gamma r+\frac12 \frac{\lambda^2\beta\gamma}{1-\gamma}\Big)+c  \sum_{i=1}^n g_{y_i} (z-x_iy_i) + cg_z \Big( \kappa(\theta-z)+\frac{\lambda\gamma\sigma\rho\sqrt{z\beta}}{1-\gamma}\Big)+\frac12 \sigma^2 czg_{zz}.
\end{equation*}
For this PDE the conditions (A1), (A2) and (A3') in \cite{HS00} are satisfied (see also Theorem \ref{theo:FK}) which implies the existence of a classical solution (see Theorem 1 in \cite{HS00}). Note that the finiteness condition (A3e') may only be satisfied on a certain time interval $[0,T_\infty]$. 

Let us consider the case $\rho=0$ in more detail. In this case $c=1$ and $f=g$. The remaining PDE is given by
\begin{eqnarray}\label{eq:pdef}
\nonumber  0 &=& f_t +\frac12 f\gamma \frac{\lambda^2 \beta}{1-\gamma} +f\gamma r + \sum_{i=1}^n f_{y_i} (z-x_iy_i) +f_z  \kappa (\theta-z) + \frac12 \sigma^2 z f_{zz}.
\end{eqnarray}
Here we use the Ansatz
\begin{equation}\label{eq:f}
f(t,y_1,\ldots,y_n,z) = \exp\Big(\phi(T-t)+\sum_{i=1}^n \psi_i(T-t)y_i+\varphi(T-t)z \Big)
\end{equation}
with $\phi(0)=\psi_i(0)=\varphi(0)=0$. Note that this approach is typical for affine models. It has already be shown to be successful in a number of stochastic volatility models (see e.g. \cite{BL13,KMK10}). We obtain:
\begin{eqnarray}
\nonumber   0 &=& -(\phi_t+ \sum_{i=1}^n y_i \psi_{it}+z\varphi_t)+\gamma r +\frac12\gamma \frac{\lambda^2 \beta}{1-\gamma} \\
   && + \sum_{i=1}^n \psi_i (z-x_iy_i ) +  \varphi \kappa(\theta-z) + \frac12 \sigma^2 z \varphi^2.
\end{eqnarray}
Inserting $\beta=v_0+ \sum_{i=1}^n y_i q_i$, rearranging the terms and using the abbreviation $\eta := \frac12\frac{\gamma\lambda^2}{1-\gamma}$  we arrive at
\begin{eqnarray}
\nonumber   0 &=& -\phi_t + \gamma r +v_0\eta+\varphi\kappa\theta \\
\nonumber   &&+ \sum_{i=1}^n y_i  \Big( -\psi_{it}+ \eta q_i - \psi_ix_i\Big)\\
   &&+z\Big(-\varphi_t + \sum_{i=1}^n \psi_i -  \kappa\varphi  + \frac12 \sigma^2  \varphi^2\Big).
\end{eqnarray}
Since this equation has to be satisfied for all $z$ and $q_i$ we end up with the following system of ordinary differential equations where $i=1,\ldots,n$:
\begin{eqnarray}\label{eq:odesystem}
\nonumber  \psi_{it}(T-t) &=&  \eta q_i - x_i\psi_i(T-t)\\
   \varphi_t(T-t) &=& \sum_{i=1}^n  \psi_i (T-t)-\kappa\varphi(T-t)+\nonumber \frac12\sigma^2 \varphi^2(T-t) \\
    \phi_t(T-t) &=& \gamma r+v_0\eta+\varphi (T-t)\kappa\theta.
\end{eqnarray}
The first differential equations for $\psi_i$ are just linear and an explicit solution together with the boundary condition $\psi_i(0)=0$ is given by
\begin{equation}\psi_i(T-t)=  \eta \frac{q_i}{x_i}(1-e^{-x_i(T-t)})=  \eta q_i \int_0^{T-t}e^{-x_is}ds.  \end{equation}
Thus we obtain
\begin{eqnarray}
\nonumber  \sum_{i=1}^n  \psi_i (T-t) &=& \eta  \sum_{i=1}^n  q_i \int_0^{T-t}e^{-x_is}ds\\
&=& \eta   \int_0^{T-t} \int_0^\infty e^{-xs} \mu^{n}(dx)ds.
\end{eqnarray}
Hence the remaining differential equations  can be written as
\begin{eqnarray}
 \varphi_t(T-t) &=& \eta \int_0^{T-t} \int_0^\infty e^{-xs} \mu^{n}(dx)ds-\kappa\varphi(T-t)+\frac12\sigma^2 \varphi^2(T-t) \label{eq:odesystem2_i}\\
  \phi_t(T-t) &=& \gamma r+v_0\eta+\varphi (T-t)\kappa\theta.\label{eq:odesystem2_ii}
\end{eqnarray} 
with boundary condition $\varphi(0)=\phi(0)=0$. Once $\varphi$ is known, the solution of \eqref{eq:odesystem2_ii} is immediate.
The existence of a solution is thus satisfied when the ordinary differential equation for $\varphi$ has a solution. This however, is guaranteed by the existence theorem of Picard-Lindel\"of due to continuity of the coefficients. Note however that the existence may only be guaranteed on a finite time interval $[0,T_\infty]$.   \hfill $\square$

{\em Proof of Theorem \ref{theo:valuefiniteneg}:}
Suppose that $G=\frac1\gamma w^\gamma g^c$ is as stated. Then it is by definition a solution of the HJB equation. To obtain $V = G$, we show that for an arbitrary investment strategy $\pi$ we have that
 \begin{align}\label{eq:ver1}
\Eop_{t,y_1,\ldots,y_n,z}\left[\frac{\left(W_T^{\pi}\right)^\gamma}{\gamma} \right] \leq G(t,w,y_1,\ldots,y_n,z)\,,
\end{align}
and that for the optimal policy $\pi^\star$ we have that
\begin{align}\label{eq:ver2}
\Eop_{t,y_1,\ldots,y_n,z} \left[ \frac{\left(W_T^{\pi^\star}\right)^\gamma}{\gamma} \right] = G(t,w,y_1,\ldots,y_n,z)\,.
\end{align}
Since $G \in C^{1,2}$, we obtain by It\^{o}'s formula for an admissible investment strategy $\pi$ that (we write $W$ instead of $W^\pi$ for simplicity)
\begin{align}
\label{eq:veri1}& G(T,W_T,Y_T^{x_1},\ldots,Y_T^{x_n},Z_T) = G(t,w,y_1,\ldots,y_n,z) + \int_t^T G_{w} W_s \pi_s \sqrt{\nu_s^n} \,dB_s^S \\  \nonumber
&+ \int_t^T G_{z} \sigma \sqrt{Z_s} \,dB_s^Z+ \int_t^T\Big\{ G_t + G_w W_s \left(r + \pi_s\lambda\nu_s^n\right) + G_z\kappa(\theta-Z_s) \\\nonumber
&+\sum_{i=1}^{n} G_{y_i} (Z_s-x_i Y_s^{x_i})+ \frac{1}{2} G_{ww} W_s^2 \pi_s^2\nu_s^n + \frac{1}{2} G_{zz} Z_s \sigma^2 +G_{wz} W_s \pi_s \sigma\rho \sqrt{Z_s \nu_s^n}\Big\} \,ds\\\nonumber
&\leq G(t,w,y_1,\ldots,y_n,z) + \int_t^T G_{w} W_s \pi_s \sqrt{\nu_s^n} \,dB_s^S + \int_t^T G_{z} \sigma \sqrt{Z_s} \,dB_s^Z\,,
\end{align}
where we have used \eqref{eq:HJB1} to obtain the expression in the last line. The right-hand side is thus a local martingale in $T$. For $\gamma > 0$ we have $G > 0$, and the right-hand side is a supermartingale, such that, using $G(T,W_T,Y_T^{x_1},\ldots,Y_T^{x_n},Z_T) = w^\gamma/\gamma$, we get  
\begin{align*}\label{eq:ver3}
\Eop_{t,y_1,\ldots,y_n,z} \left[ G(T,W_T,Y_T^{x_1},\ldots,Y_T^{x_n},Z_T) \right] = \Eop_{t,y_1,\ldots,y_n,z} \left[ \frac{\left(W_T^{\pi}\right)^\gamma}{\gamma} \right]\leq G(t,w,y_1,\ldots,y_n,z)\,,
\end{align*}
i.e.~\eqref{eq:ver1} holds. Relation~\eqref{eq:ver1} also holds for bounded admissible strategies and $\gamma < 0$, since then
\begin{align*}
\Eop_{t,y_1,\ldots,y_n,z} \left[ \int_t^T G_{w} W_s \pi_s \sqrt{\nu_s^n} \,dB_s^S \right] = \Eop_{t,y_1,\ldots,y_n,z} \left[ \int_t^T G_{z} \sigma \sqrt{Z_s} \,dB_s^Z \right] = 0\,.
\end{align*}
Relation~\eqref{eq:ver2} is obtained as follows: When we plug in $(\pi_t^*)$ into equation \eqref{eq:veri1}  we obtain by the definition of $(\pi_t^*)$ that
\begin{small}
\begin{align*}
& G(T,W_T,Y_T^{x_1},\ldots,Y_T^{x_n},Z_T) =  G(t,w,y_1,\ldots,y_n,z) + \int_t^T G_{w} W_s \pi_s^* \sqrt{\nu_s^n} \,dB_s^S + \int_t^T G_{z} \sigma \sqrt{Z_s} \,dB_s^Z.
\end{align*}
\end{small}
Taking the conditional expectation $\Eop_{t,y_1,\ldots,y_n,z}$ on both sides and making use of the assumption implies the statement.  \hfill $\square$



\begin{lemma}\label{lem:comp}
Let $(\tilde{Z}_t^n)$ be given as in  \eqref{eq:tildeZ}. If $\rho\le 0$ we have
$$ \Pop\Big(\tilde{Z}_t^n \ge \tilde{Z}_t^{n+1}, \forall t \ge 0\Big)=1.$$ In particular $\Pop\Big(\tilde{Z}_t^n \le {Z}_t, \forall t \ge 0\Big)=1$, where $(Z_t)$ is given in \eqref{eq:Z}. If $\rho\ge 0$ we have
 $$ \Pop\Big(\tilde{Z}_t^n \le \tilde{Z}_t^{n+1}, \forall t \ge 0\Big)=1.$$ 
Finally the statements for $(\tilde{Z}_t^n)$  carry over to $(\tilde{\nu}_t)$.
\end{lemma}

\begin{proof}
The proof follows the same way as the proof of Theorem 1.1 in \cite{Y79}. Note here that $(\tilde{Z}_t^n)$ can never become zero and that the function $\rho$ in the proof can be chosen as $\rho(x)=\sqrt{x}$.  Due to Lemma \ref{lem:conv1} we have monotonicity with respect to $n$ in the drift term of $(\tilde{Z}_t^n)$.
\end{proof}

{\em Proof of Lemma \ref{lem:weakconv}:} We use Proposition 5.1 in \cite{KP91}. It is first possible to show that $(\tilde{Z}_t^n)$ is relatively compact. According to \cite{KD} Theorem 2.1 we have to show that
\begin{itemize}
\item[(i)] $\lim_{K\to\infty} \sup_n \Pop(|\tilde{Z}_t^n|>K)=0, \forall 0\le t\le T.$
\item[(ii)] $\lim_{h\downarrow 0} \limsup_{n\to\infty} \sup_{t\le T} \Eop[\min\{1, |\tilde{Z}_{t+h}^n-\tilde{Z}_{t}^n|\}]=0.$
\end{itemize}
But this is true since $(\tilde{Z}_t^n)$ can be bounded by $(Z_t)$ in case $\rho\le 0$ and by $(\tilde{Z}_t)$ in case $\rho\ge 0$.
Further  we write 
$$ d\tilde{Z}_t^n = (F^n_1(\tilde{Z}^n), F^n_2(\tilde{Z}^n))d \Big( \begin{array}{c}t\\B_t^Z
\end{array}\Big)$$ 
where 
\begin{eqnarray}
F^n_1(z)_t &=& \kappa(\theta-{z}_t) + \frac{\lambda\gamma\sigma\rho}{1-\gamma} \sqrt{{z}_t} \sqrt{v_0\! +\! \sum_1^n q_i^n \int_0^t e^{-(t-s)x_i}{z}_s^nds}\\
F_1(z)_t &=& \kappa(\theta-{z}_t) + \frac{\lambda\gamma\sigma\rho}{1-\gamma} \sqrt{{z}_t} \sqrt{v_0\! +\! \int_0^\infty \int_0^t e^{-(t-s)x}{z}_sds\mu(dx)} \\
F^n_2(z)_t &=& \sigma\sqrt{z_t}\\
F_2(z)_t &=& \sigma\sqrt{z_t}.
\end{eqnarray}
We next have to show that when $(z^n)\to (z)$ in Skorohod topology, then also $(z^n, F^n(z^n)) \to (z,F(z))$ in Skorohod topology. By since $z$ and thus also $F(z)$ are continuous, this boils down to uniform convergence on compact intervals.  It remains to show that $F^n_1(z^n) \to F_1(z)$ uniformly on compact intervals if $(z^n)\to (z)$. This can be done by looking at the expressions under the square roots:
\begin{eqnarray}
&& \Big| \sum_{i=1}^n q_i^n \int_0^t e^{-(t-s)x_i}{z}_s^nds- \int_0^\infty \int_0^t e^{-(t-s)x}{z}_sds\mu(dx) \Big| \\
&\le & \sum_{i=1}^n q_i^n \int_0^t e^{-(t-s)x_i}| {z}_s^n -z_s| ds+  \int_0^t  \Big|\sum_{i=1}^n q_i^ne^{-(t-s)x_i} -  \int_0^\infty  e^{-(t-s)x}\mu(dx) \Big| {z}_sds\\
&\le& \hspace{-0.35cm} \int_0^\infty \int_0^t e^{-(t-s)x}| {z}_s^n -z_s| ds\mu(dx) +  \int_0^t  \Big|\sum_{i=1}^n q_i^ne^{-(t-s)x_i} -  \int_0^\infty  e^{-(t-s)x}\mu(dx) \Big| {z}_sds
\end{eqnarray}
The last term converges to zero due to Lemma \ref{lem:conv1} and first term because $ |{z}_s^n -z_s| \to 0$ u.o.c.\ by assumption. 

Finally the statement follows from Proposition 5.1 in \cite{KP91}.  \hfill $\square$

{\em Proof of Lemma \ref{lem:conVpi}:} From Theorem \ref{theo:conv} and the definition of the stochastic It\^{o}-integral we obtain since
$$ \lim_{n\to \infty} \Eop \Big[\int_0^T \pi_s^2 (\sqrt{\nu_s}-\sqrt{\nu_s^n})^2ds \Big]=0$$
by monotone convergence that 
$$ \int_0^T \pi_s \sqrt{\nu_s^n}dB_s^S \stackrel{L^2}{\to } \int_0^T \pi_s \sqrt{\nu_s}dB_s^S$$
which implies $L^1$-convergence. Altogether we obtain
\begin{eqnarray*}
\nonumber &&\lim_{n\to \infty}  \Big(\int_0^T  \Big(r+\pi_s \nu_s^n(\lambda-\frac12 \pi_s)\Big)ds + \int_0^T \pi_s \sqrt{\nu_s^n}dB_s^S\Big)\\
&=& \int_0^T  \Big(r+\pi_s \nu_s(\lambda-\frac12 \pi_s)\Big)ds + \lim_{n\to \infty}  \int_0^T \pi_s \sqrt{\nu_s^n}dB_s^S.
\end{eqnarray*}
Since $L^1$ convergence implies by the Skorokhod representation theorem almost sure convergence on a suitable probability space, we get
$$  \lim_{n\to \infty} W_T^{\pi,n}= W_T^\pi.$$
The uniform convergence then implies the statement.  \hfill $\square$

{\em Proof of Lemma \ref{lem:odeconvergence} }
 Let us first write these two differential equations as
 \begin{eqnarray}
 \varphi_t^n(T-t) &=& \eta h^n(T-t)-\kappa\varphi^n(T-t)+\frac12\sigma^2 (\varphi^n)^2(T-t) \\
\varphi_t(T-t) &=& \eta h(T-t)-\kappa\varphi(T-t)+\frac12\sigma^2 \varphi^2(T-t) 
\end{eqnarray} 
with $h^n(T-t)= \int_0^{T-t} \int_0^\infty e^{-xs} \mu^n(dx)ds$ and $h(T-t)=\frac{(T-t)^\alpha}{\Gamma(\alpha+1)} $. Due to Lemma \ref{lem:conv1} and the previous calculation we know that $0\le h^n(t)\le h(t)$ and $\lim_{n\to\infty}h^n(t) = h(t)$. Now let $\varphi^0$ be a solution of 
$$ \varphi_t^0(T-t) = -\kappa\varphi^0(T-t)+\frac12\sigma^2 (\varphi^0)^2(T-t)$$
with $\varphi^0(0)=0$.Then by a classical comparison theorem (see e.g. Proposition 5.2.18 in \cite{KS}) we obtain that $\varphi^0\le \varphi^n\le \varphi$ pointwise for all $n$, i.e. $\varphi^n$ is bounded for fixed $t$. Now consider $\xi^n(t) := \varphi(t)-\varphi^n(t)$. It solves the differential equation
$$ \xi^n_t(t) = \eta\big( h(t)-h^n(t)\big)-\kappa\xi^n(t)+\frac12\sigma^2\xi^n(t)\big( \varphi_n(t)+\varphi(t)\big).$$
and is hence given by
$$ \xi^n(t) = \eta\exp\Big( \int_0^t \big(\frac12\sigma^2(\varphi^n(u)+\varphi(u))-\kappa\big)du\Big)\Big( \int_0^t e^{\int_0^u \big(-\frac12\sigma^2 (\varphi^n(s)+\varphi(s))+\kappa\big) ds}\big(h(u)-h^n(u)\big)du\Big).$$
Since $\varphi^n$ is bounded and $\lim_{n\to\infty}h^n(t) = h(t)$ we obtain by dominated convergence that $\xi^n(t) \to 0$ for $n\to\infty$ which implies the stated convergence.
  \hfill $\square$

{\em Proof of Theorem \ref{theo:limitcase} }
Since $\pi^*\equiv \frac{\lambda}{1-\gamma}$ is optimal for the $n$-th approximation by Theorem \ref{theo:valuefinite} we obtain
\begin{equation}\label{eq:ineq}
\Eop_{w_0,v_0,z_0}\Big[\frac1\gamma (W_T^{\pi,n})^\gamma\Big] \le \Eop_{w_0,v_0,z_0}\Big[\frac1\gamma (W_T^{*,n})^\gamma\Big]
\end{equation}
where we write $W_T^{*,n}$ instead of $W_T^{\pi^*,n}$. Let us now consider the expectation on the right-hand side. Inserting the optimal portfolio strategy $\pi^*$ and using the Feynman-Kac result in Theorem \ref{theo:FK} yields
\begin{eqnarray*}
 \Eop_{w_0,v_0,z_0}\Big[\frac1\gamma (W_T^{*,n})^\gamma\Big] 
= \frac1\gamma w_0^\gamma  \exp(\gamma rT)\Eop_{w_0,v_0,z_0}\Big[ \exp\Big( \frac{\lambda^2\gamma}{2(1-\gamma)}  \int_0^T  \nu_s^n ds\Big) \Big].
\end{eqnarray*}
By monotone convergence we obtain from Theorem \ref{theo:conv} that
$$\lim_{n\to \infty} \Eop_{w_0,v_0,z_0}\Big[ \exp\Big( \frac{\lambda^2\gamma}{2(1-\gamma)}  \int_0^T  \nu_s^n ds\Big) \Big] = \Eop_{w_0,v_0,z_0}\Big[ \exp\Big( \frac{\lambda^2\gamma}{2(1-\gamma)}  \int_0^T  \nu_s ds\Big) \Big] .$$
On the other hand we know that
$$ \exp(\gamma rT)\Eop_{w_0,v_0,z_0}\Big[ \exp\Big( \frac{\lambda^2\gamma}{2(1-\gamma)}  \int_0^T  \nu_s^n ds\Big) \Big] = \exp\Big(\phi^n(T)+\varphi^n(T)z_0 \Big)
$$
where $\varphi^n$ and $\phi^n$ are solutions of \eqref{riccati1} and \eqref{ode2}.  From Lemma \ref{lem:odeconvergence} we know that the right-hand side converges to $$\exp\Big(\phi(T)+\varphi(T)z_0 \Big)
$$
where $\varphi$ and $\phi$ are solutions of \eqref{riccati2} and \eqref{ode3}. Thus we get that
\begin{eqnarray*}
\lim_{n\to\infty }\Eop_{w_0,v_0,z_0}\Big[\frac1\gamma (W_T^{*,n})^\gamma\Big] &=& \frac1\gamma w_0^\gamma  \exp(\gamma rT)\Eop_{w_0,v_0,z_0}\Big[ \exp\Big( \frac{\lambda^2\gamma}{2(1-\gamma)}  \int_0^T  \nu_s ds\Big) \Big]\\ &=& \frac1\gamma w_0^\gamma \exp\Big(\phi(T)+\varphi(T)z_0 \Big).
\end{eqnarray*}    Now we have to consider the left-hand side of \eqref{eq:ineq}.

From the proof of Lemma  \ref{lem:conv1} we know that for any $\pi$
$$  \lim_{n\to \infty} W_T^{\pi,n}= W_T^\pi.$$
Now from Fatou's Lemma we finally obtain (in case $\gamma\ge 0$ we use $0$ as a lower bound, in case $\gamma<0$ we have assume uniform integrability)
$$ \liminf_{n\to\infty} \Eop_{w_0,v_0,z_0}\Big[\frac1\gamma (W_T^{\pi,n})^\gamma\Big] \ge \Eop_{w_0,v_0,z_0}\Big[\liminf_{n\to\infty} \frac1\gamma (W_T^{\pi,n})^\gamma\Big] = \Eop_{w_0,v_0,z_0}\Big[\frac1\gamma (W_T^{\pi})^\gamma\Big]  .$$
This implies now that for all admissible policies $(\pi_t)$
\begin{eqnarray*}
\Eop_{w_0,v_0,z_0}\Big[\frac1\gamma (W_T^{\pi})^\gamma\Big] &\le& 
\frac1\gamma w_0^\gamma  \exp(\gamma rT)\Eop_{w_0,v_0,z_0}\Big[ \exp\Big( \frac{\lambda^2\gamma}{2(1-\gamma)}  \int_0^T  \nu_sds\Big) \Big].
\end{eqnarray*}
What is left to show is that the upper bound can be obtained. 
Inserting the candidate for the optimal portfolio strategy $\pi^*\equiv \frac{\lambda}{1-\gamma}$ in the SDE of the wealth equation  yields
\begin{small}
\begin{eqnarray*}
 \Eop_{w_0,v_0,z_0}\Big[\frac1\gamma (W_T^{\pi^{\star}})^\gamma\Big] &=& 
\frac1\gamma w_0^\gamma \Eop_{w_0,v_0,z_0}\Big[\exp\Big(\int_0^T  \gamma\Big(r+\frac{\lambda^2(1-2\gamma)}{2(1-\gamma)^2} \nu_s\Big)ds + \int_0^T \frac{\lambda\gamma}{1-\gamma}\sqrt{\nu_s}dB_s^S\Big)\Big]\\
&=& \frac1\gamma w_0^\gamma \Eop_{w_0,v_0,z_0}\Big[\exp\Big(\gamma rT + \frac{\lambda^2\gamma}{2(1-\gamma)}  \int_0^T  \nu_s ds\Big) M_T\Big]
\end{eqnarray*}
\end{small}
with 
$$M_T = \exp\Big( -\int_0^T \frac12\frac{\lambda^2\gamma^2}{(1-\gamma)^2}\nu_sds+ \int_0^T \frac{\lambda\gamma}{1-\gamma}\sqrt{\nu_s}dB_s^S\Big).$$
Let $(\FC_t^Z)$ be the filtration which is generated by $(B_t^Z)$ only. Then since $(B_t^Z)$ and $(B_t^S)$ are uncorrelated we obtain by conditioning and since $\Eop[M_T]=1$ (this follows like in Example 4 in \cite{LS01}) that
\begin{small}
 \begin{eqnarray*}
 \Eop_{w_0,v_0,z_0}\Big[\frac1\gamma (W_T^{\pi^{\star}})^\gamma\Big] &=& 
\frac1\gamma w_0^\gamma  \Eop_{w_0,v_0,z_0}\Big[ \Eop_{w_0,v_0,z_0}\Big[\exp\Big(\gamma rT + \frac{\lambda^2\gamma}{2(1-\gamma)}  \int_0^T  \nu_s ds\Big) M_T\Big| \FC_T^Z\Big]\Big]\\
&=& \frac1\gamma w_0^\gamma  \exp(\gamma rT)\Eop_{w_0,v_0,z_0}\Big[ \exp\Big( \frac{\lambda^2\gamma}{2(1-\gamma)}  \int_0^T  \nu_s ds\Big) \Big] 
\end{eqnarray*}
\end{small}
which implies the statement.
 \hfill $\square$

{\em Proof of Theorem \ref{theo:limitcaseneg} }
We consider the approximation first. For an arbitrary admissible portfolio strategy $\pi$ we obtain
\begin{eqnarray}\label{eq:comp_rough}
  W_T^{\pi,n} &=& w_0 \exp\Big(\int_0^T  \Big(r+\pi_s a(\nu_s^n)(\lambda-\frac12 \pi_s)\Big)ds + \int_0^T \pi_s \sqrt{a(\nu_s^n)}dB_s^S\Big).
\end{eqnarray}
Note that $\pi$ is $(\FC_t)-$adapted and thus independent of $n$.
Since $\pi^*\equiv \frac{\lambda}{1-\gamma}$ is optimal for the $n$-th approximation by Theorem \ref{theo:valuefiniteneg} we obtain
\begin{equation}\label{eq:ineq_rough}
\Eop_{w_0,v_0,z_0}\Big[\frac1\gamma (W_T^{\pi,n})^\gamma\Big] \le \Eop_{w_0,v_0,z_0}\Big[\frac1\gamma (W_T^{*,n})^\gamma\Big]
\end{equation}
where we write $W_T^{*,n}$ instead of $W_T^{\pi^*,n}$.  Inserting the optimal portfolio strategy $\pi^*$ and using the Feynman-Kac result in Theorem \ref{theo:FKneg} yields for the rigth-hand side
\begin{eqnarray}
 \Eop_{w_0,v_0,z_0}\Big[\frac1\gamma (W_T^{*,n})^\gamma\Big] 
= \frac1\gamma w_0^\gamma  \exp(\gamma rT)\Eop_{w_0,v_0,z_0}\Big[ \exp\Big( \frac{\lambda^2\gamma}{2(1-\gamma)}  \int_0^T  a(\nu_s^n)ds\Big) \Big].
\end{eqnarray}
By monotone convergence we then get
\begin{equation}
\lim_{n\to \infty} \Eop_{w_0,v_0,z_0}\Big[ \exp\Big( \frac{\lambda^2\gamma}{2(1-\gamma)}  \int_0^T  a(\nu_s^n) ds\Big) \Big] = \Eop_{w_0,v_0,z_0}\Big[ \exp\Big( \frac{\lambda^2\gamma}{2(1-\gamma)}  \int_0^T  a(\nu_s) ds\Big) \Big] .
\end{equation}
From  the definition of the stochastic It\^{o}-integral we obtain since
\begin{equation}
 \lim_{n\to \infty} \Eop \Big[\int_0^T \pi_s^2 (\sqrt{a(\nu_s)}-\sqrt{a(\nu_s^n)})^2ds \Big]=0
\end{equation}
by monotone convergence that 
\begin{equation} \int_0^T \pi_s \sqrt{a(\nu_s^n)}dB_s^S \stackrel{L^2}{\to } \int_0^T \pi_s \sqrt{a(\nu_s)}dB_s^S\end{equation}
which implies $L^1$-convergence. Altogether we obtain
\begin{eqnarray}
\nonumber &&\lim_{n\to \infty}  \Big(\int_0^T  \Big(r+\pi_s a(\nu_s^n)(\lambda-\frac12 \pi_s)\Big)ds + \int_0^T \pi_s \sqrt{a(\nu_s^n)}dB_s^S\Big)\\
&=& \int_0^T  \Big(r+\pi_s a(\nu_s)(\lambda-\frac12 \pi_s)\Big)ds + \lim_{n\to \infty}  \int_0^T \pi_s \sqrt{a(\nu_s^n)}dB_s^S.
\end{eqnarray}
Since $L^1$ convergence implies by the Skorokhod representation theorem almost sure convergence on a suitable probability space, we get
\begin{equation}
 \lim_{n\to \infty} W_T^{\pi,n}= W_T^\pi.
\end{equation}
Now from Fatou's Lemma we finally obtain
\begin{equation}
 \liminf_{n\to\infty} \Eop_{w_0,v_0,z_0}\Big[\frac1\gamma (W_T^{\pi,n})^\gamma\Big] \ge \Eop_{w_0,v_0,z_0}\Big[\liminf_{n\to\infty} \frac1\gamma (W_T^{\pi,n})^\gamma\Big] = \Eop_{w_0,v_0,z_0}\Big[\frac1\gamma (W_T^{\pi})^\gamma\Big] .
\end{equation}
This implies now that for all admissible policies $(\pi_t)$
\begin{eqnarray}
\Eop_{w_0,v_0,z_0}\Big[\frac1\gamma (W_T^{\pi})^\gamma\Big] &\le& 
\frac1\gamma w_0^\gamma  \exp(\gamma rT)\Eop_{w_0,v_0,z_0}\Big[ \exp\Big( \frac{\lambda^2\gamma}{2(1-\gamma)}  \int_0^T  a(\nu_s)ds\Big) \Big].
\end{eqnarray}
Finally, the value on the right-hand side is exactly the value function of the portfolio strategy $\pi_t^*\equiv \frac{\lambda}{1-\gamma}$.
 \hfill $\square$

\section*{Acknowledgements}
S. Desmettre is grateful for support within the DFG-Research Training Group 1932 \textit{Stochastic Models for Innovations in the Engineering Sciences}. S. Desmettre is also supported by the Austrian Science Fund (FWF) project F5508-N26, which is part of the Special Research Program \textit{Quasi-Monte Carlo Methods: Theory and Applications”}. Moreover, the authors gratefully acknowledge support from NAWI Graz. 


\bibliographystyle{abbrv}

\end{document}